\documentclass[11pt,letter]{article}
\usepackage{amsmath}
\usepackage{xspace}
\usepackage{todonotes}
\usepackage{amsthm}
\usepackage{fullpage}

\usepackage[utf8]{inputenc}
\usepackage{indentfirst,verbatim} 
\usepackage{listings} 
\usepackage[normalem]{ulem}
\usepackage{soul}
\usepackage{amsfonts}
\usepackage{amsthm}
\usepackage{enumitem}
\usepackage{hyperref}
\usepackage{mdframed}

\usepackage{amsmath,amsfonts,amsthm}
\usepackage{mathbbol}
\usepackage{amsthm}
\usepackage{graphicx,color}
\usepackage{boxedminipage}
\usepackage[linesnumbered, boxed]{algorithm2e}
\usepackage{framed}
\usepackage{thmtools}
\usepackage{thm-restate}
\usepackage{xspace}
\usepackage{todonotes}
\usepackage{pgfplots}
\usetikzlibrary{calc}
\usetikzlibrary{shapes}
\usetikzlibrary{arrows.meta}
 \usepackage{colortbl}
 \usepackage{tcolorbox} 

\usepackage{xifthen}
\usepackage{tabularx}
\usepackage{capt-of}
\usepackage{thmtools}
\usepackage{thm-restate}
\usepackage{subcaption}
\usepackage{booktabs}
\usepackage{calc}
\usepackage{cleveref}

\makeatletter
\renewcommand{\theAlgoLine}{%
	\@arabic{\numexpr\value{algocf}+1\relax}.\arabic{AlgoLine}}
\makeatother

\usepackage[margin=1in]{geometry} 

\theoremstyle{plain}
\newtheorem{theorem}{Theorem}

\newtheorem{corollary}{Corollary}
\newtheorem{lemma}{Lemma}
\newtheorem{claim}{Claim}

\theoremstyle{definition}
\newtheorem{definition}{Definition}

\newcommand{\pname}{\textsc}
\newcommand{\polyn}{n^{\Oh(1)}}

\newcommand{\claimqed}{$\lrcorner$}
\newenvironment{claimproof}{\medskip\noindent \emph{Proof of Claim~\theclaim.}  }{\hfill\claimqed\medskip}
\newcommand{\pathlength}{L}
\newcommand{\cyclelength}{L}
\newcommand{\mathpathlength}{$\pathlength$\xspace}
\newcommand{\mathcyclelength}{$\cyclelength$\xspace}

\newlength{\RoundedBoxWidth}
\newsavebox{\GrayRoundedBox}
\newenvironment{GrayBox}[1]%
   {\setlength{\RoundedBoxWidth}{.93\textwidth}
    \def\boxheading{#1}
    \begin{lrbox}{\GrayRoundedBox}
       \begin{minipage}{\RoundedBoxWidth}}%
   {   \end{minipage}
    \end{lrbox}
    \begin{center}
    \begin{tikzpicture}%
       \node(Text)[draw=black!20,fill=white,rounded corners,%
             inner sep=2ex,text width=\RoundedBoxWidth]%
             {\usebox{\GrayRoundedBox}};
        \coordinate(x) at (current bounding box.north west);
        \node [draw=white,rectangle,inner sep=3pt,anchor=north west,fill=white] 
        at ($(x)+(6pt,.75em)$) {\boxheading};
    \end{tikzpicture}
    \end{center}}     

\newenvironment{defproblemx}[2][]{\noindent\ignorespaces%
                                \FrameSep=6pt%
                                \parindent=0pt%
                \vspace*{-1.5em}
                \ifthenelse{\isempty{#1}}{%
                  \begin{GrayBox}{#2}%
                }{%
                  \begin{GrayBox}{#2 parameterized by~{#1}}%
                }
                \begin{tabular*}{\textwidth}{@{\hspace{.1em}} >{\itshape} p{1.8cm} p{0.8\textwidth} @{}}%
            }{
                \end{tabular*}%
                \end{GrayBox}%
                \ignorespacesafterend
            }  

\newcommand{\Oh}{\mathcal{O}}

\newcommand{\dg}{{\rm dg}}

\newcommand{\dist}{{\rm dist}}

\newcommand{\probDC}{\pname{Long Dirac  Cycle}\xspace}

\newcommand{\probKCycle}{\pname{Longest Cycle}\xspace}

	\newcommand{\bananadec}{Erd{\H {o}}s-Gallai decomposition\xspace}
\newcommand{\banana}{Erd{\H {o}}s-Gallai component\xspace}
\newcommand{\bananas}{Erd{\H {o}}s-Gallai components\xspace}
	\newcommand{\Bananadec}{Erd{\H {o}}s-Gallai decomposition\xspace}

\newcommand{\cyclebananadec}{Dirac decomposition\xspace}
\newcommand{\cyclebanana}{Dirac component\xspace}
\newcommand{\Cyclebananadec}{Dirac decomposition\xspace}
\newcommand{\Cyclebanana}{Dirac component\xspace}
\newcommand{\Nestedbananadec}{Nested \bananadec}
\newcommand{\nestedbananadec}{nested \bananadec}
\newcommand{\LongNestedPathName}{long\_nested\_st\_path}
\newcommand{\NestedCompressionName}{nested\_compress}
\newcommand{\NestedDecompressionName}{nested\_decompress}

\lstnewenvironment{code}{\lstset{language=Haskell,basicstyle=\small}}{}

\title{Approximating Long Cycle Above Dirac’s Guarantee}

\author{
Fedor V. Fomin\thanks{
Department of Informatics, University of Bergen, Norway.}\\fomin@ii.uib.no
\and
Petr A. Golovach\addtocounter{footnote}{-1}\footnotemark{}\\petr.golovach@ii.uib.no
\and
Danil Sagunov\thanks{
   St.\ Petersburg Department of V.A.\ Steklov Institute of Mathematics, Russia.
}\\danilka.pro@gmail.com
\and 
Kirill Simonov\thanks{Hasso Plattner Institute, University of Potsdam, Germany}\\kirillsimonov@gmail.com
}

\usetikzlibrary{arrows}

\begin{document}

\maketitle

\begin{abstract}
Parameterization above (or below) a guarantee is a successful concept in parameterized algorithms.  
The idea is that many computational problems admit ``natural''  guarantees
bringing to algorithmic questions whether a better solution (above the guarantee)  could be obtained efficiently. For example, for every boolean CNF formula on $m$ clauses, there is an assignment that satisfies at least $m/2$ clauses. How difficult is it to decide whether there is an assignment satisfying more than $m/2 +k$ clauses?  
  Or, if an $n$-vertex graph has a perfect matching, then its vertex cover is at least $n/2$. Is there a vertex cover of size at least $n/2 +k$ for some $k\geq 1$ and how difficult is it to find such a vertex cover? 
  
The above guarantee paradigm has led to several exciting discoveries in the areas of parameterized algorithms and kernelization. We argue that this paradigm could bring forth fresh perspectives on well-studied problems in approximation algorithms. Our example is the longest cycle problem. One of the oldest results in extremal combinatorics is the celebrated Dirac's theorem from 1952. Dirac's theorem provides the following guarantee on the length of the longest cycle:  for every 2-connected $n$-vertex graph $G$ with minimum degree $\delta(G)\leq n/2$,   the length of a longest cycle $L$ is  at least $2\delta(G)$. Thus the ``essential'' part in finding the longest cycle is in approximating the ``offset'' $k = L - 2 \delta(G)$. The main result of this paper is the above-guarantee approximation theorem for  $k$. 
Informally, the theorem says that approximating the offset $k$ is not harder than approximating the total length $L$ of a cycle. In other words, 
for any (reasonably well-behaved) function $f$, a polynomial time algorithm constructing a cycle 
of length $ f(L)$ in an undirected graph with a cycle of length $L$, yields a polynomial time algorithm constructing a cycle of length 
$2\delta(G)+\Omega(f(k))$.
\end{abstract}

\section{Introduction}
One of the concepts that had a strong impact on the development of parameterized algorithms and kernelization is the idea of the above guarantee parameterization. Above guarantee parameterization  grounds on the following observation: 
\emph{the natural parameterization of a maximization/minimization problem by the solution size is not satisfactory if there is a lower bound for the solution size that is sufficiently large} \cite{FominGLPSZ21}. To make this discussion concrete, consider the example of the classical NP-complete problem \textsc{Max Cut}. Observe that in any graph with $m$ edges there is always a cut containing at least  $m/2$ edges. (Actually,  slightly better bounds are known in the literature  \cite{MR0337686,BollobasScot02}.) 
Thus \textsc{Max Cut} is trivially fixed-parameter tractable (FPT) parameterized by the size of the max-cut. Indeed, the following simple algorithm shows that the problem is FPT: 
 If~$k\leq m/2$, then return {\sf yes}; else~$m\leq 2k$  and any brute-force algorithm will do the job. However, the question about \textsc{Max Cut} becomes much more meaningful and interesting, when one seeks a cut above the ``guaranteed'' lower bound $m/2$. 
 
 The above guarantee approach was introduced by Mahajan and Raman~\cite{MahajanRS09} and it was successfully applied in the study of several fundamental problems in parameterized complexity and kernelization. For illustrative examples, we refer to~\cite{AlonGKSY10,BezakovaCDF17,CrowstonJMPRS13,FominGLPSZ21,FominGSS22esa,GargP16,DBLP:journals/mst/GutinKLM11,GutinIMY12,GutinP16,GutinRSY07,LokshtanovNRRS14}, see also the recent survey of Gutin and Mnich \cite{GutinMnich22}.  Quite surprisingly, the theory of the above (or below) guarantee \emph{approximation} remains unexplored. (Notable exceptions are the works of Mishra et al. \cite{MishraRSSS11} on approximating the minimum vertex cover beyond the size of a maximum matching
  and of Bollob\'{a}s and Scott on approximating max-cut beyond the $m/2 +\sqrt{m/8}$ bound \cite{BollobasScot02}.)

In this paper, we bring the philosophy of the above guarantee parameterization into the realm of approximation algorithms. 
 In particular,
 \begin{quote}
 \begin{mdframed}[backgroundcolor=gray!20]  
 The goal of this paper is to study the approximability of the classical problems of finding a longest cycle and a longest $(s,t)$-path in a graph from the viewpoint of the above guarantee parameterization. 
    \end{mdframed}
\end{quote}

\subparagraph{Our results.}  
 Approximating the length of a longest cycle in a graph enjoys a lengthy and rich history~\cite{BjorkHus03,BjorklundHK04,FederMS02,FederMotw10,FurerR94,Gabow07,Vishwanathan04}. 
There are several fundamental results in extremal combinatorics providing lower bounds on the length of a longest cycle in a graph. The oldest of these bounds is given by   
 Dirac's Theorem from  1952 
\cite{Dirac52}. Dirac's Theorem 
states that a $2$-connected graph $G$ with the minimum vertex degree $\delta(G)$ contains a cycle of length $\cyclelength\geq \min\{2\delta(G),|V(G)|\}$.
Since  every longest cycle in a graph $G$ with $\delta(G)<\frac{1}{2}|V(G)|$ (otherwise, $G$ is Hamiltonian and a longest cycle can be found in polynomial time) always has a   ``complementary'' part of length $2\delta(G)$, the 
essence of the problem is in  computing the ``offset'' $k=\cyclelength-2\delta(G)$. 
Informally, the first main finding of our paper is that Dirac's theorem is well-compatible with approximation.
We prove that approximating the offset $k$ is essentially not more difficult than approximating the length  $\cyclelength$. 

More precisely. Recall that $f$ is subadditive if for all $x$, $y$ it holds that $f(x + y) \le f(x) + f(y)$. Our  main result is the following theorem.

\begin{restatable}{theorem}{diracappxtheorem}\label{thm:diracappxtheorem}\label{thm:long_cycle_approx}
 	Let $f: \mathbb{R}_+\to \mathbb{R}_+$ be a non-decreasing subadditive function and suppose that we are given a polynomial-time algorithm finding a cycle of length at least $f(\cyclelength)$ in graphs  with 
	the longest cycle length  \mathcyclelength. Then there exists a polynomial time algorithm that finds a cycle of length at least $2\delta(G)+\Omega(f(\cyclelength-2\delta(G)))$ in a $2$-connected graph $G$ with $\delta(G) \le \frac{1}{2}|V(G)|$ and the longest cycle length $L$.
\end{restatable}

The  2-connectivity condition is important. As was noted in ~\cite{FominGSS22},   deciding whether a connected graph $G$ contains a cycle of length at least $2\delta(G)$ is NP-complete.
\Cref{thm:diracappxtheorem} trivially extends to approximating the longest path problem above $2\delta(G)$. 
For the longest path, the requirement on 2-connectivity of a graph can be relaxed to connectivity. 
This can be done by a standard reduction of adding an apex vertex $v$  to the connected graph $G$, see e.g.~\cite{FominGSS22}.    The minimum vertex degree in the new graph $G+v$, which is 2-connected,  is equal to $\delta(G)+1$, and $G$ has a path of length at least $
\pathlength$ if and only if $G+v$ has a cycle of length at least $\pathlength+2$.  Thus approximation of the longest cycle (by making use of \Cref{thm:diracappxtheorem}) in $G+v$, is also the approximation of the longest path in $G$.

\subparagraph{Related work.}   The first approximation algorithms for longest paths and cycles followed the development of exact parameterized algorithms. Monien \cite{Monien85}  and Bodlaender  \cite{Bodlaender93a}  gave parameterized algorithms computing a path of length $\pathlength$ in times 
 $\Oh(\pathlength!2^\pathlength n)$ and  $\Oh(\pathlength! n m)$ respectively.  These algorithms imply also  approximation algorithms constructing in polynomial time  a path of length $  \Omega(\log{\pathlength}/\log\log{\pathlength)}$, where  $\pathlength$ is the longest path length in graph $G$. In their celebrated work on color coding,  Alon, Yuster, and, Zwick~\cite{AlonYZ95} obtained an algorithm 
  that in time   $\Oh(5.44^\pathlength n)$   finds   a path/cycle of length $\pathlength$. The algorithm of Alon et al. implies constructing in polynomial time a path of   length $  \Omega(\log{\pathlength})$. 
 A significant amount of the consecutive work targets to improve the base of the exponent $c^\pathlength$ in the running times of the parameterized algorithms for longest paths and cycles \cite{Koutis08,Williams09,FominLS14,DBLP:journals/siamcomp/Bjorklund14,BjHuKK10}. The surveys    \cite{FominK13,KoutisW16}, and 
\cite[Chapter~10]{cygan2015parameterized} provide an overview of  ideas and methods in this research direction.
The exponential dependence in $\pathlength$ in the running times of these algorithms is   
  asymptotically optimal: An algorithm finding a path (or cycle) of length $\pathlength$ in time  $2^{o(\pathlength)}n^{\Oh(1)}$ would fail   the Exponential Time Hypothesis (ETH) of Impagliazzo, Paturi,  and Zane~\cite{ImpagliazzoPZ01}.
  Thus none of the further improvements in the running times of parameterized algorithms for longest cycle or path, would lead to a better than   $  \Omega(\log{\pathlength})$ approximation bound.  
  
Bj\"{o}rklund and Husfeldt \cite{BjorkHus03} made the first step ``beyond color-coding'' in approximating the longest path. They gave a polynomial-time algorithm that finds a path of length $\Omega(\log \pathlength/ \log \log \pathlength)^2$ in a graph with the longest path length $\pathlength$. Gabow in \cite{Gabow07} enhanced and extended this result to  approximating the longest cycle. His algorithm computes a cycle of length $2^{\Omega(\sqrt{\log \cyclelength/\log \log \cyclelength})}$ in a graph with a cycle of length $\cyclelength$. Gabow and Nie \cite{GabowNisaac08} observed that  a refinement of Gabow's algorithm leads to a polynomial-time algorithm constructing cycles of length  $2^{\Omega(\sqrt{\log \cyclelength})}$.
This is better than $(\log (L))^{\Oh(1)}$  but worse than $L^\varepsilon$.
Pipelining  the algorithm of Gabow and Nie  with \Cref{thm:diracappxtheorem} yields a polynomial time algorithm constructing in a $2$-connected graph $G$ a cycle of length  $2\delta(G) +\Omega(c^{\sqrt{\log{k}}})$.
For graphs of bounded vertex degrees, better approximation algorithms are known~\cite{chen2006approximating,FederMS02}.

The gap between the upper and lower bounds for the longest path approximation is still big. Karger, Motwani,  and Ramkumar~\cite{KargerMR97} proved that 
the longest path problem does not belong to APX unless P = NP. They also show that for any $\varepsilon >0$,  it cannot be approximated within  $2^{\log^{1-\varepsilon} n}$ unless NP $\subseteq$  DTIME($2^{O(\log^{1/\varepsilon} n)}$). Bazgan, Santha, and Tuza~\cite{BazganST99} extended these lower bounds to cubic Hamiltonian graphs.  
For directed graphs the gap between the upper and lower bounds is narrower \cite{BjorklundHK04,GabowN08}.

Our approximation algorithms are inspired by the recent work Fomin, Golovach, Sagunov, and Simonov~\cite{FominGSS22} on the parameterized complexity of the longest cycle beyond Dirac's bound. 
  Fomin et al. were interested in computing the ``offset'' beyond $2\delta(G)$ exactly. Their parameterizes algorithm decides whether $G$ contains a cycle of length at least $2\delta(G) +k$ in 
  time $2^{\Oh(k)}n^{\Oh(1)} $, and thus in   
 polynomial time computes a cycle of length   $2\delta(G) +\Omega(\log{k})$. However, the tools developed 
 in~\cite{FominGSS22} are not sufficient to go beyond $\Omega(\log{k})$-bound on the offset. The main combinatorial tools from ~\cite{FominGSS22} are \bananadec  and \cyclebananadec of graphs. For the needs of approximation, we have to develop novel  (``nested'') variants or prove additional structural properties of these decompositions.

  Dirac's theorem is one of the central pillars of  Extremal Graph Theory. The excellent surveys \cite{MR1373656} and \cite{MR1373679} provide an introduction to this fundamental subarea of graph theory. Besides \cite{FominGSS22}, the algorithmic applications of Dirac's theorem from the perspective of parameterized complexity were studied by  
   Jansen, Kozma, and Nederlof  in \cite{DBLP:conf/wg/Jansen0N19}. 

   \subparagraph{Paper structure.}  \Cref{sec:overview} provides an overview of the techniques employed to achieve our results. Then, \Cref{sec:prelim} introduces notations and lists auxiliary results.
   \Cref{sec:eg} guides through the proof of the approximation result for $(s,t)$-paths, which is the key ingredient required for \Cref{thm:long_cycle_approx}. \Cref{sec:cycle} is dedicated to the proof of \Cref{thm:long_cycle_approx} itself. Finally, we conclude with a summary and some open questions in \Cref{sec:concl}.

\section{Overview of the proofs}\label{sec:overview}
In this section, we provide a high-level strategy of the proof of \Cref{thm:long_cycle_approx}, as well as key technical ideas needed along the way.
The central concept of our work is an approximation algorithm for the \textsc{Longest Cycle} problem.
Formally, such an algorithm should run in polynomial time for a given graph $G$ and should output a cycle of length at least $f(\cyclelength)$, where $\cyclelength$ is the length of the longest cycle in $G$.
The function $f$ here is the \emph{approximation guarantee} of the algorithm.
In our work, we allow it to be an arbitrary non-decreasing function $f: \mathbb{R}_+\to \mathbb{R}_+$ that is also subadditive (i.e., $f(x)+f(y)\ge f(x+y)$ for arbitrary $x,y$).
We also note that an $f(\cyclelength)$-approximation algorithm for \textsc{Longest Cycle} immediately gives a $\frac{1}{2}f(2\cyclelength)$-approximation algorithm for \textsc{Long $(s,t)$-Path} in $2$-connected graphs (by Menger's theorem, see \Cref{lemma:cycle_to_path} for details).

Our two main contributions assume that we are given such an $f$-approximation algorithm as a black box.
In fact, we only require to run this algorithm on an arbitrary graph as an oracle and receive its output.
We do not need to modify or know the algorithm routine.

While the basis of our algorithm comes from the structural results of Fomin et al.~\cite{FominGSS22}, in the first part of this section we do not provide the details on how it is used.

The first of our contributions is a polynomial-time algorithm that finds a long $(s,t)$-path in a given $2$-connected graph $G$ with  two vertices $s,t\in V(G)$.
The longest $(s,t)$-path in $G$ always has length $\delta(G-\{s,t\})+k$ for $k\ge 0$ by Erd\H{o}s-Gallai theorem, and the goal of the algorithm is to find an $(s,t)$-path of length at least $\delta(G-\{s,t\})+\Omega(f(k))$ in $G$.
To find such a path, this algorithm first recursively decomposes the graph $G$ in a specific technical way.
As a result, it outputs several triples $(H_i, s_i, t_i)$ in polynomial time, where $H_i$ is a $2$-connected minor of $G$ and $s_i,t_i \in V(H_i)$.
For each triple, the  algorithm runs the black box to find a $f$-approximation of the longest $(s_i,t_i)$-path in $H_i$.
In the second round, our algorithm cleverly uses constructed approximations to construct a path of length at least $\delta(G-\{s,t\})+\Omega(f(k))$ in the initial graph $G$.
This is summarized as the following theorem.

\begin{restatable}{theorem}{erdgallappxtheorem}\label{thm:eg_approx}\label{thm:erdgallaiappxtheorem}
	Let $f: \mathbb{R}_+\to \mathbb{R}_+$ be a non-decreasing subadditive function and suppose that we are given a polynomial-time algorithm computing an $(s,t)$-path of length at least $f(\pathlength)$ in graphs with given two vertices $s$ and $t$ having  the longest $(s,t)$-path of length  $\pathlength$. Then 	
	there is a polynomial-time algorithm that outputs an $(s,t)$-path of length at least  $\delta(G-\{s,t\})+\Omega(f(\pathlength-\delta(G-\{s,t\})))$ in a $2$-connected graph $G$ with two given vertices $s$ and $t$ having the longest $(s,t)$-path length $\pathlength$. 
\end{restatable}

The second (and main) contribution of this paper is the polynomial-time algorithm that approximates the longest cycle in a  given $2$-connected graph $G$  such that $2\delta(G) \le |V(G)|$.
It employs the black-box $f$-approximation algorithm for \probKCycle to find a cycle of length $2\delta(G)+\Omega(f(k))$, where $2\delta(G)+k$ is the length of the longest cycle in $G$.
By Dirac's theorem applied to $G$, $k$ is always at least $0$.

To achieve that, our algorithm first tries to decompose the graph $G$. However, in contrast to the first contributed algorithm, here the decomposition process is much simpler.
In fact, the decomposition routine is never applied recursively, as the decomposition itself needs not to be used: its existence is sufficient to apply another, simple, procedure.

Similarly to the first contribution, the algorithm then outputs a series of triples $(H_i, s_i, t_i)$, where $H_i$ is a $2$-connected minor of $G$ and $s_i, t_i \in V(H_i)$.
The difference here is that for each triple the algorithm runs not the initial black-box $f$-approximation algorithm, but the algorithm of the first contribution, i.e.\ the algorithm of \Cref{thm:erdgallaiappxtheorem}.
Thus, the output of each run is an $(s_i,t_i)$-path of length $\delta(H_i-\{s_i,t_i\})+\Omega(f(k_i))$ in $H_i$, where $\delta(H_i-\{s_i,t_i\})+k_i$ is the length of the longest $(s_i,t_i)$-path in $H_i$.

Finally, from each approximation, our algorithm constructs a cycle of length at least $2\delta(G)+\Omega(f(k_i))$.
It is guaranteed that $k_i=\Omega(k)$ for at least one $i$, so the longest of all constructed cycles is of length at least $2\delta(G)+\Omega(f(k))$.
The following theorem is in order.

\diracappxtheorem*

One may note that \Cref{thm:erdgallaiappxtheorem} actually follows from \Cref{thm:diracappxtheorem} (again, by Menger's theorem, see \Cref{lemma:cycle_to_path}).
However, as described above, the algorithm in \Cref{thm:diracappxtheorem} employs the algorithm of \Cref{thm:erdgallaiappxtheorem}, so we have to prove the latter before the former.

In the remaining part of this section, we provide more detailed proof overviews of both theorems, in particular, we explain how the algorithms employ the structural results of~\cite{FominGSS22}.
In both proofs, we complement these results by showing useful properties of specific graph decompositions.
For clarity, we start with \Cref{thm:diracappxtheorem}, as its proof is less involved.

\subsection{Approximating long cycles}
The basis of our algorithm is the structural result due to Fomin et al.~\cite{FominGSS22}.
In that work, the authors show the following: There is an algorithm that, given a cycle in a $2$-connected graph, either finds a longer cycle or finds that $G$ is of a ``particular structure''.
This algorithm can be applied to any cycle of length less than $(2+\sigma_1)\cdot \delta(G)$ (to be specific, we use $\sigma_1=\frac{1}{24}$, see  \Cref{lemma:main_cycle_lemma} for details).

To see how this stuctural result is important, recall that we aim to find a cycle of length at least $2\delta(G)+\Omega(f(k))$ in a $2$-connected graph $G$ with the longest cycle length $2\delta(G)+k$.
Our algorithm simply starts with some cycle in $G$ and applies the result of~\cite{FominGSS22} to enlarge it exhaustively.
It stops when either a cycle is of length at least $(2+\sigma_1)\cdot \delta(G)$, or the particular structure of $G$ is found.

The crucial observation here is that if a long cycle is found, we can trivially find a good approximation.
If $\sigma_1 \cdot \delta(G)$ is, e.g., less than $\sigma_1/10 \cdot f(k)$, then $10\delta(G)<f(k)$.
If we just apply the blackbox $f$-approximation algorithm for the \probKCycle problem, we get a cycle of length at least $f(2\delta(G)+k)\ge f(k)\ge 2\delta(G)+4/5 \cdot f(k)$.
Hence, by taking the longest of the cycles of length $(2+\sigma_1)\cdot f(k)$ and of length $f(2\delta(G)+k)$ we always achieve a good approximation guarantee on $k$.

The most important part of the algorithm is employed when the ``particular structure'' outcome is received from the structural lemma applied on $G$ and the current cycle $C$.
Here we need to be specific about this structure, and the outcome can be of two types.
The first outcome is a bounded vertex cover of the graph.
This vertex cover is of size at most $\delta(G)+2(k'+1)$, where $k'\ge 0$ is such that $|V(C)|=2\delta(G)+k'$.
Such vertex cover is a guarantee that $C$ is not much shorter than the longest cycle in $G$: the length of the longest cycle is bounded by twice the vertex cover size, so $k\le 4(k'+1)$.
Hence, $k'=\Omega(k)$ and $C$ is a sufficient approximation.

The second, and last, structural outcome is the Dirac decomposition, defined in \cite{FominGSS22}.
Basically, this decomposition is obtained by finding a small separator of $G$ (that consists of just two subpaths $P_1,P_2$ of the cycle $C$), and the parts of this decomposition are the connected components of $G$ after the separation.
The main result on Dirac decomposition proved in \cite{FominGSS22} is that there always exists a longest cycle that contains an edge in at least one of these parts.

While the definition and properties of Dirac decomposition may seem quite involved, our algorithm does not even require the Dirac decomposition of $G$ to be found.
In fact, we show a new nice property of Dirac decomposition.
It guarantees that if a Dirac decomposition for $G$ exists, then there also exists a $2$-vertex separator $\{u,v\}$ of $G$ that also divides the longest cycle in $G$ into almost even parts.
Our contribution is formulated in the following lemma.

\begin{restatable}{lemma}{lemsinglepath}\label{lemma:single_path_lemma}
	Let $G$ be a $2$-connected graph and $P_1, P_2$ induce a \cyclebananadec  for a cycle $C$ of length at most $2\delta(G)+\kappa$ in $G$ such that $2\kappa \le \delta(G)$.
	If there exists a cycle of length at least $2\delta(G)+k$ in $G$, then there exist $u,v \in V(G)$ such that
	\begin{itemize}
		\item $G-\{u,v\}$ is not connected, and
		\item there is an $(u,v)$-path of length at least $\delta(G)+(k-2)/4$ in $G$.
	\end{itemize} 
\end{restatable}

Our algorithm employs Lemma~\ref{lemma:single_path_lemma} in the following way.
Since there are $\Oh(|V(G)|^2)$ vertex pairs in $G$, our algorithm iterates over all vertex pairs.
If a pair $u, v$ separates the graph into at least two parts, then our algorithm finds a long $(u,v)$-path that contains vertices in only one of the parts.
Formally, it iterates over all connected components in $G-\{u,v\}$.
For a fixed connected component $H$, our algorithm applies the algorithm of \Cref{thm:erdgallaiappxtheorem} to the graph $G[V(H)\cup \{u,v\}]+uv$ (the edge $uv$ is added to ensure $2$-connectivity), to find an approximation of the longest $(u,v)$-path.
By \Cref{lemma:single_path_lemma}, if $u,v$ is the required separating pair, then for at least one $H$ the length of the found $(u,v)$-path should be $\delta(G)+\Omega(k)$.
And if such a path is found, a sufficiently long $(u,v)$-path outside $H$ in $G$ is guaranteed by Erd\H{o}s-Gallai theorem.
Together, these two paths form the required cycle of length $2\delta(G)+\Omega(k)$.

With that, the proof overview of \Cref{thm:diracappxtheorem} is finished.
The formal proof is present in \Cref{sec:cycle}.

\subsection{Approximating long $(s,t)$-paths}

While the algorithm of \Cref{thm:diracappxtheorem} does not use the underlying Dirac decomposition explicitly, in the case of finding $(s,t)$-paths (and to prove \Cref{thm:erdgallaiappxtheorem}), we require deeper usage of the obtained graph decomposition.
While the Dirac decomposition of Fomin et al.\ was originally used in \cite{FominGSS22} to find long cycles above $2\delta(G)$, for finding $(s,t)$-paths above $\delta(G-\{s,t\})$ the authors introduced the Erd\H{o}s-Gallai decomposition.

In the formal proof of \Cref{thm:erdgallaiappxtheorem} in \Cref{sec:eg}, we give a complete definition of Erd\H{o}s-Gallai decomposition.
In this overview, we aim to avoid the most technical details in order to provide an intuition of the structure of the decomposition and how our algorithm employs it. 

Similarly to Dirac decomposition, the Erd\H{o}s-Gallai decomposition is obtained through the routine that, given a graph $G$ and an $(s,t)$-path inside it, either enlarges the path or reports that two subpaths $P_1$ (that starts with $s$) and $P_2$ (that starts with $t$) of the given path induce (when deleted) an Erd\H{o}s-Gallai decomposition in $G$.
This routine can be applied to an $(s,t)$-path until it reaches $(1+\sigma_2)\cdot \delta(G-\{s,t\})$ in length (specifically, $\sigma_2=\frac{1}{4}$, see \Cref{lemma:st_path_or_tunnel}; in this overview, we also skip the case of a Hamiltonian $(s,t)$-path for brevity).
Note that, in contrast to the cycle enlargement routine of the Dirac decomposition, here the bounded vertex cover outcome is not possible.
Similarly to the algorithm of the previous subsection, the only non-trivial part of the algorithm is dealing with the Erd\H{o}s-Gallai decomposition outcome.
In the other case, a single run of the black-box $f$-approximation algorithm for \probKCycle provides the desired approximation immediately.

The main property of this decomposition due to \cite{FominGSS22} is as follows: If an $(s,t)$-path of length at least $\delta(G-\{s,t\})+k$ exists in $G$, then there necessarily exists the path of length at least $\delta(G-\{s,t\})+k$ that goes through one of the connected components in the decomposition.
Moreover, for each of the connected components $G_i$ there is exactly one pair of distinct entrypoints $s_i, t_i$: if an $(s,t)$-path in $G$ goes through $G_i$, it should necessary enter $G_i$ in $s_i$ (or $t_i$) once and leave $G_i$ in $t_i$ (or $s_i$) exactly once as well.

Additionally to that, we have that the degree of each $G_i$ is not much different from $G$: $\delta(G_i-\{s_i,t_i\})\ge \delta(G-\{s,t\})-2$ holds true.
And this constant difference is always compensated by paths from $s$ and $t$ to $s_i$ and $t_i$: if we succeed to find an $(s_i,t_i)$-path of length at least $\delta(G_i-\{s_i,t_i\})+k_i$ inside $G_i$, we can always complete it with \emph{any} pair of disjoint paths from $\{s,t\}$ to $\{s_i,t_i\}$ into an $(s,t)$-path of length $\delta(G-\{s,t\})+k_i$ in $G$.
Should this pair be longer than the trivial lower bound of $2$, it grants the additional length above $\delta(G-\{s,t\})+k_i$.

The previous paragraph suggests the following approach for our approximation algorithm: for each $G_i,s_i,t_i$, our algorithm applies itself recursively to find an $(s_i,t_i)$-path of length $\delta(G_i-\{s_i,t_i\}+\Omega(f(k_i))$, where $k_i$ comes from the longest $(s_i,t_i)$-path length in $G_i$.
Since the other part of the additional length comes from two disjoint paths between $\{s,t\}$, and $\{s_i,t_i\}$, we would like to employ the black-box $f$-approximation algorithm to find the $f$-approximation of this pair of paths.

Unfortunately, finding such pair of paths reduces only to finding a long cycle through a given pair of vertices (it is enough to glue $s$ with $t$ and $s_i$ with $t_i$ in $G$, and ask to find the long cycle through the resulting pair of vertices).
In their work, Fomin et al.\ have shown that the problem of finding such a cycle of length at least $k$ can be done in $2^{\Oh(k)}\cdot \polyn$ time.
However, this is of little use to us, as $k$ is only bounded by $\Oh(\delta(G))$, but we require polynomial time.
Simultaneously, we do not know of any way to force the black-box algorithm to find an $f$-approximation for a cycle through the given pair of vertices.

These arguments bring us away from the idea of a recursive approximation algorithm.
Instead, our approximation algorithm will apply the black-box algorithm to a single ``complete-picture'' graph that is obtained according to the structure brought by the Erd\H{o}s-Gallai decomposition.
However, the recursion here remains in the sense that we apply the path-enlarging routine to each component of the decomposition.
This brings us to the idea of the recursive decomposition, which we define as the \emph{nested Erd\H{o}s-Gallai decomposition} in \Cref{sec:eg}.
This decomposition can be seen as a tree, where the root is the initial triple $(G,s,t)$, the children of a node represent the triples $(G_i, s_i,t_i)$ given by the Erd\H{o}s-Gallai decomposition, and the leaves of this decomposition are the graphs $G_i$ where sufficient approximations of long $(s_i,t_i)$-paths are found (by taking the longest of $(1+\sigma_2)\cdot \delta(G-\{s_i,t_i\})$-long path from the enlarging routine and the approximation obtained from the blackbox algorithm).
A schematic picture of this novel decomposition is present in \Cref{fig:nested_banana_example}.

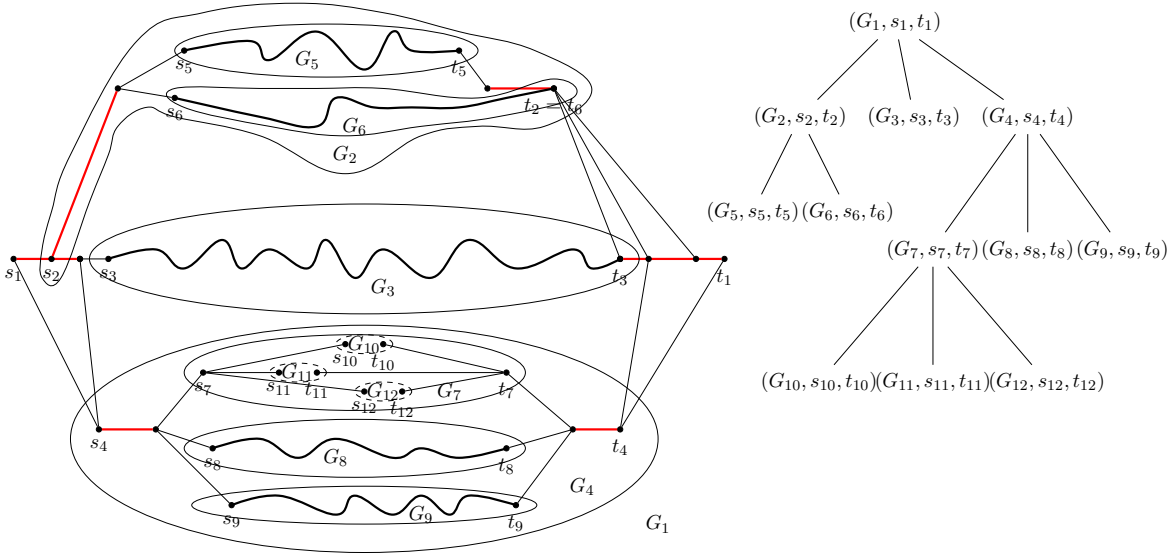
\begin{figure}[ht]
	\centering
	\scalebox{0.7}{\begin{tikzpicture}[scale=0.36]
	\tikzstyle{vertex}=[draw,circle,fill,inner sep=0pt,minimum size=0.1cm]
	\tikzstyle{edge}=[]
	\tikzstyle{longpath}=[very thick]
	\tikzstyle{suffixprefix}=[red,very thick,edge]
	\tikzstyle{borders}=[thin]
	\tikzstyle{smalldegree}=[dashed]
	\tikzstyle{treenode}=[]
	\tikzstyle{treeedge}=[]
	
\node [vertex,label=below:$s_1$] (s) at (-5.5,0.5) {};
\node [vertex] (a) at (-2,0.5) {};
\node [vertex,label=below:$s_3$] (s2) at (-0.5,0.5) {};
\node [vertex] (sp) at (a) {};
\draw[suffixprefix] (s) -- (sp);
\node [vertex,label=below:$t_1$] (t) at (32,0.5) {};
\node [vertex] (tp) at (26.5,0.5) {};

\node [vertex,label=below:$s_2$] (s1) at (-3.5,0.5) {};
\node [vertex,label=below:{$t_3$}] (t2) at (tp) {};

\draw[suffixprefix] (t) -- (tp);
\draw[edge] (sp)--(s2);

\draw[borders]  plot[smooth cycle, tension=.7] coordinates {(-3.5,3)  (-4,1) (-4,-0.5)  (-3,-0.5) (-1.5,3.5) (-0.5,6) (0,7.5)  (1.5,8.5) (3,8) (7.5,7) (12,5) (16,7) (19,7.5) (22.5,7.5)       (25,9.5) (22,11.5) (18,12.5) (14,13.5) (9.5,14) (4.5,13.5) (1,12) (-2.5,6) };
\node [vertex,label=below:{$t_2=t_6$}] (t1) at (23,9.5) {};
\node [vertex] (b) at (28,0.5) {};
\draw[edge] (t1) -- (b);
\node [vertex] (sp1) at (0,9.5) {};
\draw[suffixprefix] (s1) -- (sp1);
\draw[borders]  (13,0.5) ellipse (14.5 and 2.9);

\node [vertex,label=below:$s_4$] (s3) at (-1,-8.5) {};
\draw[edge] (a) -- (s3);
\draw[edge] (s) -- (s3);
\draw[longpath]  plot[smooth, tension=.7] coordinates {(s2) (2,1) (3.5,0) (5.5,1.5) (6.5,0) (8,1) (10,0) (11,1.5) (12.5,-0.5) (14,1) (16.5,-0.5) (18.5,1.5) (21,0) (24,1) (25,0) (tp)};

\node [vertex,label=below:$t_4$] (t3) at (26.5,-8.5) {};
\draw[edge] (t) -- (t3);

\node [vertex] (d) at (30.5,0.5) {};

\draw[edge] (d) -- (t1);
\draw[borders]  (13,-9) ellipse (15.5 and 6);
\node [vertex] (sp3) at (2,-8.5) {};
\draw[suffixprefix] (s3) -- (sp3);

\node [vertex] (tp3) at (24,-8.5) {};
\draw[suffixprefix] (t3) -- (tp3);
\node [vertex] (tp1) at (19.5,9.5) {};
\draw[suffixprefix] (t1) -- (tp1);
\draw[borders]  (11,11.5) ellipse (8 and 1.4);
\node [vertex,label=below:$s_5$] (s4) at (3.5,11.5) {};
\node [vertex,label=below:$t_5$] (t4) at (18,11.5) {};
\draw  (sp1) edge (s4);
\draw  (tp1) edge (t4);
\draw[longpath]  plot[smooth, tension=.7] coordinates {(s4) (7,12) (8,11) (10.5,12.5) (13,10.5) (14.5,12.5) (15.5,11.5) (t4)};
\draw[borders]  plot[smooth cycle, tension=.7] coordinates {(19,9) (23,10) (24,9) (20,8) (16,7.5) (12,7) (8,7.5) (3.5,8.5) (3,9.5) (7.5,9.5) (10.5,9.5) (14,9.5)  };
\node[vertex] (t5) at (t1) {};
\node [vertex,label=below:$s_6$] (s5) at (3,9) {};
\draw [edge] (sp1) -- (s5);
\draw [longpath] plot[smooth, tension=.7] coordinates {(s5) (7.5,8) (10.5,7.5) (11.5,9) (13.5,8.5) (15.5,8.5) (18,8.5) (20.5,9) (t1)};
\draw  [borders](12.5,-5.5) ellipse (9 and 2);
\draw[borders]  (12.5,-9.5) ellipse (9 and 1.53);
\draw [borders] (13,-12.5) ellipse (9.1029 and 1);
\node [vertex,label=below:$s_7$] (s6) at (4.5,-5.5) {};
\node [vertex,label=below:$t_7$] (t6) at (20.5,-5.5) {};
\node [vertex,label=below:$s_8$] (s7) at (5,-9.5) {};
\node [vertex,label=below:$t_8$] (t7) at (20.5,-9.5) {};
\node [vertex,label=below:$s_9$] (s8) at (6,-12.5) {};
\node [vertex,label=below:$t_9$] (t8) at (21,-12.5) {};
\draw[borders,smalldegree]  (13,-4) ellipse (1.5 and 0.5);
\draw[borders,smalldegree]  (14,-6.5) ellipse (1.5 and 0.5);
\draw[borders,smalldegree]  (9.5,-5.5) ellipse (1.5 and 0.5);
\draw[longpath]  plot[smooth, tension=.7] coordinates {(s7) (7.5,-9) (9.5,-10) (11.5,-9) (14.5,-10) (16,-9.5) (18,-10) (t7)};
\draw[longpath]  plot[smooth, tension=.7] coordinates {(t8) (18.5,-12) (17.5,-13) (16,-12) (14.5,-13) (12.5,-12) (11.5,-13) (8.5,-12) (s8)};
\node [vertex,label=below:$s_{10}$] (s9) at (12,-4) {};
\node [vertex,label=below:$t_{10}$] (t9) at (14,-4) {};
\node [vertex,label=below:$s_{11}$] (s10) at (8.5,-5.5) {};
\node [vertex,label=below:$t_{11}$] (t10) at (10.5,-5.5) {};
\node [vertex,label=below:$s_{12}$] (s11) at (13,-6.5) {};
\node [vertex,label=below:$t_{12}$] (t11) at (15,-6.5) {};
\draw[edge] (s6) -- (s9);
\draw[edge] (s6) -- (s10);
\draw[edge] (s6) -- (s11);
\draw[edge] (t6) -- (t9);
\draw[edge] (t6) -- (t10);
\draw[edge] (t6) -- (t11);
\draw[edge] (sp3) -- (s6);
\draw[edge] (sp3) -- (s7);
\draw[edge] (sp3) -- (s8);
\draw[edge] (tp3) -- (t6);
\draw[edge] (tp3) -- (t7);
\draw[edge] (tp3) -- (t8);

\node at (10,11) {$G_5$};
\node at (12.5,7.5) {$G_6$};
\node at (14,-1) {$G_3$};
\node at (17.5,-6.5) {$G_7$};
\node at (14,-6.5) {$G_{12}$};
\node at (9.5,-5.5) {$G_{11}$};
\node at (13,-4) {$G_{10}$};
\node at (11.5,-10) {$G_8$};
\node at (16,-13) {$G_9$};
\node at (12,6) {$G_2$};
\node at (24.5,-11.5) {$G_4$};
\draw[edge] (tp) -- (t1);
\draw[edge] (b) -- (t3);
\node at (28.5,-13.5) {$G_1$};
\begin{scope}[shift={(33,31)}]
\node [treenode] (T1) at (8,-18) {$(G_1,s_1,t_1)$};
\node [treenode] (T2) at (3,-23) {$(G_2,s_2,t_2)$};
\node [treenode] (T3) at (9,-23) {$(G_3,s_3,t_3)$};
\node [treenode] (T4) at (15,-23) {$(G_4,s_4,t_4)$};
\draw[treeedge] (T1)--(T2);
\draw[treeedge] (T1)--(T3);
\draw[treeedge] (T1)--(T4);
\node [treenode] (T5) at (0.5,-28) {$(G_5,s_5,t_5)$};
\node [treenode] (T6) at (5.5,-28) {$(G_6,s_6,t_6)$};
\draw[treeedge] (T2) -- (T5);
\draw[treeedge] (T2) -- (T6);
\node [treenode] (T7) at (10,-30) {$(G_7,s_7,t_7)$};
\node [treenode] (T8) at (15,-30) {$(G_8,s_8,t_8)$};
\node [treenode] (T9) at (20,-30) {$(G_9,s_9,t_9)$};
\draw[treeedge] (T4) -- (T7);
\draw[treeedge] (T4) -- (T8);
\draw[treeedge] (T4) -- (T9);

\node [treenode] (T10) at (4,-37) {$(G_{10},s_{10},t_{10})$};
\node [treenode] (T11) at (10,-37) {$(G_{11},s_{11},t_{11})$};
\node [treenode] (T12) at (16,-37) {$(G_{12},s_{12},t_{12})$};
\draw[treeedge] (T7) -- (T10);
\draw[treeedge] (T7) -- (T11);
\draw[treeedge] (T7) -- (T12);
\end{scope}
\end{tikzpicture}}
	\caption{A schematic example of a nested \bananadec (left) and the corresponding recursion tree (right).
		Red straight paths inside $G_i$ denote the pair of paths inducing an \bananadec in $G_i$.
		Bold $(s_i,t_i)$-paths are sufficient approximations of the longest $(s_i,t_i)$-paths in $G_i$.
		Dashed contours correspond to $G_i$ with constant $\delta(G_i-\{s_i,t_i\})$, which is one of a few technical cases in the proof.}\label{fig:nested_banana_example}
\end{figure}

In \Cref{sec:eg}, we show that a long path found inside a leaf $(G_i,s_i,t_i)$ of the decomposition can be contracted into a single edge $s_it_i$.
Moreover, if $(G_j,s_j,t_j)$ is a child of a $(G_i, s_i, t_i)$ in the decomposition, and the longest pair of paths from $\{s_i,t_i\}$ to $\{s_j,t_j\}$ is just a pair of edges (so it does not grant any additional length as described before), we contract these edges.
The crucial in our proof is the claim that after such a contraction, if an $(s,t)$-path of length $\delta(G-\{s,t\})+k$ exists in the initial graph, an $(s,t)$-path of length at least $\Omega(k)$ exists in the graph obtained with described contractions.
After doing all the contractions, the algorithm applies the black-box algorithm to the transformed graph and finds an $(s,t)$-path of length $f(\Omega(k))$ (which is $\Omega(f(k))$ by subadditivity) inside it.

The final part of our algorithm (and the proof of \Cref{thm:erdgallaiappxtheorem}) is the routine that \emph{transforms} this $(s,t)$-path inside the \emph{contracted} graph $G$ into a path of length $\delta(G-\{s,t\})+\Omega(f(k))$ in the \emph{initial} graph $G$.
In this part, we prove that it is always possible to transform an $(s,t)$-path of length $r$ in the contracted graph into a path of length $\Omega(r)$ that goes through at least one edge corresponding to a leaf of the nested Erd\H{o}s-Gallai decomposition (hence, to a good approximation of $(s_i,t_i)$-path inside $G_i$).
Finally, we observe that reversing the contractions in $G$ transforms this path into the required approximation.

This finishes the overview of the proof of \Cref{thm:erdgallaiappxtheorem}.
\Cref{sec:eg} contains it fully, with all technical details and formal proofs.

\section{Preliminaries}\label{sec:prelim} 
In this section, we define the notation used throughout  the paper and provide some auxiliary results. We use $[n]$ to denote the set of positive integers $\{1,\ldots,n\}$. We remind that a function 
$f\colon D\rightarrow \mathbb{R}$ is \emph{subadditive} if $f(x+y)\leq f(x)+f(y)$ for all $x,y\in D\subseteq\mathbb{R}$. We denote the set of all nonnegative real numbers by $\mathbb{R}_+$.

Recall that our main theorems are stated for arbitrary nondecreasing subadditive functions $f: \mathbb{R}_+ \to \mathbb{R}$, such that an algorithm achieving the respective approximation exists.
Throughout the proofs, we will additionally assume that $f(x)\le x$ for every $x \in \mathbb{R}_+$.
For any integer $x \ge 3$, this is already implied by the statement, since a consistent approximation algorithm cannot output an $(s, t)$-path (respectively, cycle) of length greater than $x$ in a graph where the longest $(s, t)$-path (respectively, cycle) has length $x$.
However, for a general function $f(\cdot)$ this does not necessarily hold on the whole $\mathbb{R}_+$.
If this is the case, for clarity of the proofs we redefine $f(x):=\min\{x,f(x)\}$ for every $x \in \mathbb{R}_+$.
Clearly, $f$ remains subadditive and non-decreasing, while also imposing exactly the same guarantee on the approximation algorithm.

\subparagraph{Graphs.}
We consider only finite simple undirected graphs and use the standard notation (see, e.g.,  the book of Diestel~\cite{Diestel}).
We use $V[G]$ and $E(G)$ to denote the sets of vertices and edges, respectively, of a graph $G$.
Throughout the paper, we use $n$ and $m$ to denote the number of vertices and the number of edges of a considered graph if it does not create confusion. 
For a set $X\subseteq V(G)$, $G[X]$ is used to denote the subgraph of $G$ \emph{induced} by $X$ and we write $G-X$ to denote the subgraph of $G$ induced by $V(G)\setminus X$. For a single-vertex set $\{v\}$, we write $G-v$ instead of $G-\{v\}$.   For a vertex $v$, $N_G(v)$ denotes the \emph{(open) neighborhood} of $v$, that is, the set of the neighbors of $v$ in $G$.  
For a set  set $X\subseteq V(G)$, $N_G(X)=\big(\bigcup_{v\in X}N_G(v)\big)\setminus X$. 
The \emph{degree} of a vertex $v$ is $\deg_G(v)=|N_G(v)|$. We denote by  $\delta(G)=\min_{v\in V(G)}\deg_G(v)$  the \emph{minimum degree} of $G$.
We may omit the subscript in the above notation if the considered graph is clear from the context. 
We remind that the \emph{edge contraction} operation for $uv\in E(G)$ replaces $u$ and $v$ by a single vertex $w_{uv}$ that is adjacent to every vertex of $N_G(\{u,v\})$.   
A set of vertices $X\subseteq V(G)$ is \emph{vertex cover} if every edge of $G$ has at least one endpoint in $X$.

A \emph{path} $P$ in a graph $G$ is a subgraph of $G$ whose set of vertices can be written as $\{v_0,\ldots,v_k\}$ where $E(P)=\{v_{i-1}v_i\mid  i\in[k]\}$. We may write a path $P$ as the sequence of its vertices $v_0,\ldots, v_k$. The vertices $v_0$ and $v_k$ are called \emph{endpoints} of $P$ and other vertices are \emph{internal}. For a path $P$ with endpoints $s$ and $t$, we say that $P$ is an $(s,t)$-path.   
Two paths $P_1$ and $P_2$ are \emph{(vertex-)disjoint} if they have no common vertex and \emph{internally disjoint} if no internal vertex of either of the paths is a vertex of the other path. 
A \emph{cycle} $C$ in $G$ is a subgraph of $G$ with $V(C)=\{v_1,\ldots v_k\}$ and $E(C)=\{v_{i-1}v_i\mid i\in[k]\}$, where $k\geq 3$ and it is assumed that $v_0=v_k$.
The \emph{length} of a path (a cycle, respectively) is the number of its edges. 
For two internally disjoint paths $P_1=v_0,\ldots,v_k$ and $P_2=u_0,\ldots,v_s$ sharing exactly one endpoint $v_k=u_0$, we write $P_1P_2$ to denote their \emph{concatenation}, that is, the path $v_0,\ldots,v_k,u_1,\ldots,u_s$. If $P_1$ and $P_2$ share both endpoints and at least one of them has internal vertices, we write $P_1P_2$ to denote the cycle composed by the paths. 
A path $P$ (a cycle $C$, respectively) is \emph{Hamiltonian} if $V(P)=V(G)$ ($V(C)=V(G)$, respectively).  

Recall that $G$ is \emph{connected} if for every two vertices $s$ and $t$, $G$ contains an $(s,t)$-path. A \emph{(connected) component} of $G$ is an inclusion maximal connected induced subgraph.  A connected graph $G$ with at least three vertices is \emph{$2$-connected} if for every $v\in V(G)$, $G-v$ is connected. A vertex $v$ of a connected graph $G$ with at least two vertices is a \emph{cut-vertex} if $G-v$ is disconnected. 
 A \emph{block} of a connected graph $G$ is an inclusion maximal induced subgraph without cut-vertices. Note that if $G$ has at least two vertices, then each block is either isomorphic to $K_2$ or a 2-connected graph. 
 For a block $B$ of $G$, a vertex $v\in V(B)$ that is not a cut-vertex of $G$ is called \emph{inner}. 
  Blocks in a connected graph form a tree  structure (viewing each block as a vertex of the forest and two blocks are adjacent if they share a cut-vertex). The blocks corresponding to the leaves of the block-tee, are called \emph{leaf-blocks}.  
For $s,t\in V(G)$, $S\subseteq V(G)\setminus \{s,t\}$ is an \emph{$(s,t)$-separator} if $G-S$ has no $(s,t)$-path; we also say that $S$ \emph{separates} $s$ from $t$. We also say that $S$ separates two sets of vertices $A$ and $B$ if $S$ separates each vertex of $A$ from every vertex of $B$.

\medskip
The following useful observation follows immediately from Menger's theorem (see, e.g.,~\cite{Diestel,KleinbergT06}). 

 \begin{lemma}\label{lemma:cycle_to_path}
For any $2$-connected graph $G$ with a cycle of length $\cyclelength$, there is a path of length at least $L/2$ between any pair of vertices in $G$. Moreover, given a cycle $C$ and two distinct vertices $s$ and $t$, an $(s,t)$-path of length at least $|V(C)|/2$ can be constructed in polynomial time. 
\end{lemma}

We observe that given an approximation algorithm  for a longest cycle, we can use it as a black box to approximate a longest path between any two vertices.

 \begin{lemma}\label{lemma:stpath_approx}
 Let $\mathcal{A}$ be a polynomial-time algorithm that finds a cycle of length at least $f(\cyclelength)$ in a graph with the longest cycle length $\cyclelength$.
Then there is a polynomial-time algorithm using $\mathcal{A}$ as a subroutine that, given a graph $G$ and two distinct vertices $s$ and $t$, finds an $(s,t)$-path of length at least  $\frac{1}{2}f(2\pathlength)$, where $\pathlength$  is the length of a longest $(s,t)$-path in $G$.
\end{lemma}

\begin{proof}
Let $G$ be a graph and let $s,t\in V(G)$ be distinct vertices. We assume without loss of generality that $G$ is connected. Let $P$ be a longest $(s,t)$-path in $G$ and let $\pathlength$ be its length. If $st$ is a bridge of $G$, then $G$ has a unique $(s,t)$-path and and its length is one. In this case, our algorithm returns this path that trivially can be found in polynomial time. Assume that this is not the case. Then $st\notin E(P)$ and $\pathlength\geq 2$.

We construct two copies $G_1$ and $G_2$ of $G$. Denote by $s_1$ and $s_2$ the copies of $s$ in $G_1$ and $G_2$, respectively. Similarly, let $t_1$ and $t_2$ be the copies of $t$, and denote by $P_1$ and $P_2$ the copes of $P$ in $G_1$ and $G_2$, respectively. Next, we construct the graph $G'$ by unifying $s_1$ and $s_2$, and $t_1$ and $t_2$ (if $st\in E(G)$, the edges $s_1t_1$ and $s_2t_2$ are unified as well). Denote by $s'$ the vertex of $G'$ obtained from $s_1$ and $s_2$, and let $t'$ be the vertex obtained from $t_1$ and $t_2$.  Note that $P_1$ and $P_2$ are internally disjoint  $(s',t')$-paths in $G'$. In particular, this implies that $s'$ and $t'$ are vertices of the same block $B$ of $G'$, and $P_1$ and $P_2$ are paths in $B$.  Therefore, $B$ contains the cycle $C=P_1P_2$ of length $2\pathlength$. We obtain that the longest cycle length in $B$ is at least $2\pathlength$. We call $\mathcal{A}$ on $B$ and this algorithm outputs a cycle $C$ of length at least $f(2\pathlength)$. Note that $B$ is distinct from $K_2$, i.e., is 2-connected. 
 By Lemma~\ref{lemma:cycle_to_path}, $B$ has an $(s',t')$-path $P'$ of length at least $\frac{1}{2}|V(C)|\geq\frac{1}{2}f(2\pathlength)$ that can be constructed in polynomial time.  
Notice that $\{s',t'\}$ separates $V(G_1)\setminus\{s_1,t_1\}$ from $V(G_2)\setminus \{s_2,t_2\}$. Hence, $P'$ is either an $(s_1,t_1)$-path in $G_1$ or $(s_2,t_2)$-path in $G_2$.
Assume  that $P'$ is a path in $G_1$ (the other case is symmetric). Since $G_1$ is a copy of $G$, the copy of $P'$ in $G$ is an $(s,t)$-path of length at least $\frac{1}{2}f(2\pathlength)$ as required by the lemma.

Since $G'$ can be constructed in polynomial time and the unique block $B$ of $G'$ containing $s'$ and $t'$ can be found in polynomial (linear) time (see, e.g.,~\cite{KleinbergT06}), the overall running time is polynomial. 
\end{proof}

We will use as a subroutine an algorithm finding two disjoint paths between two pairs of vertices of total length at least $k$, where $k$ is the given parameter. For us, constant values of $k$ suffice, though in fact there exists an FPT algorithm for this problem parameterized by the total length. It follows as an easy corollary from the following result of \cite{FominGSS22} about \textsc{Long ($s$, $t$)-Cycle}, the problem of finding a cycle of length at least $k$ through the given two vertices $s$ and $t$.

\begin{theorem}[Theorem 4 in \cite{FominGSS22}]
    There exists an FPT algorithm for \textsc{Long ($s$, $t$)-Cycle} parameterized by $k$.
    \label{thm:long_st_cycle}
\end{theorem}

For completeness, we show  the corollary next.
\begin{corollary}
    There is an FPT algorithm that, given a graph $G$ with two pairs of vertices $\{s, t\}$ and $\{s', t'\}$, and a parameter $k$, finds two disjoint paths between $\{s, t\}$ and $\{s', t'\}$ in $G$ of total length at least $k$, or correctly determines that such paths do not exist.
    \label{cor:two_long_st_paths}
\end{corollary}

\begin{proof}
    Construct a new graph $H$ that consists of the graph $G$ together with two additional vertices $u$ and $v$. The vertex $u$ has exactly two neighbors in $H$, $s$ and $t$, and the neighbors of $v$ are $s'$ and $t'$. Now run the algorithm for \textsc{Long ($s$, $t$)-Cycle}with the parameter $k + 4$ to find a cycle in $H$ going through the vertices $u$ and $v$. If such a cycle is found, then removing the vertices $u$ and $v$ from it yields a pair of disjoint paths between $\{s, t\}$ and $\{s', t'\}$ in $G$ of total length at least $k$. In the other direction, if there is a pair of desired disjoint paths in $G$, then together with the vertives $u$ and $v$ they constitute a cycle of length at least $k + 4$ in $H$.
\end{proof}

Finally, it is convenient to use the following corollary, which generalizes the theorem of Erd{\H{o}}s and Gallai \cite[Theorem 1.16]{ErdosG59}.

\begin{corollary}[Corollary 3 in~\cite{FominGSS22}]\label{thm:relaxed_st_path}
		Let $G$ be a $2$-connected graph and let $s, t$ be a pair of distinct vertices in $G$.
		For any $B \subseteq V(G)$ there exists a path of length at least $\delta(G- B)$ between $s$ and $t$ in $G$. Moreover, there is a polynomial time algorithm constructing a path of such length. 
	\end{corollary}

\section{Approximating $(s,t)$-path}\label{sec:eg}

In this section, we  provide the formal proof of Theorem~\ref{thm:eg_approx}, stating that any guarantee for approximating the longest cycle in a 2-connected graph can be transferred to approximating the longest $(s, t)$-path above minimum degree. For the convenience of the reader, we recall the precise statement next.

\erdgallappxtheorem*

In order to obtain this result, we first recall the concept of \bananadec introduced in~\cite{FominGSS22} together with a few of its helpful properties established there. Then we introduce the recursive generalization of this concept, called \nestedbananadec, and show how to obtain with its help the compression of the graph such that a long $(s, t)$-path in the compressed graph can be lifted to an $(s, t)$-path in the original graph with a large offset.

\subsection{\bananadec}

This subsection encompasses the properties of an \bananadec, defined next. The definition itself and most of the technical results presented here are due to~\cite{FominGSS22}. Some of the results from~\cite{FominGSS22} need to be modified in order to be used for our purposes, we supply such results with full proofs. Note that the statements in~\cite{FominGSS22} hold in the more general case where there is also a low-degree vertex subset in the graph, here while recalling the results we automatically simplify the statements. Next, we recall the definition of an \bananadec.

\begin{restatable}[\Bananadec and \banana, Definition 2 in~\cite{FominGSS22}]{definition}{defbananadec}\label{def:banana}
	Let $P$ be a path in a $2$-connected graph $G$.
	We say that  two disjoint paths $P_1$ and $P_2$ in $G$ induce \emph{an \bananadec   for   $P$}  in $G$ if
	\begin{itemize}
		\item 
		Path $P$ is of the form $P=P_1 {P'}P_2$, where the inner path ${P'}$ has at least $\delta(G-\{s,t\})$ edges.
		\item 
		There are at least two connected components in $G-V(P_1  \cup P_2)$, and for every connected component $H$, it holds that  $|V(H)|\ge 3$ and one of the following.
		\begin{enumerate}[label=(R\arabic*)]
			\item\label{enum:tunnel_path_bic} $H$ is $2$-connected and the maximum size of a matching in  $G$ between $V(H)$ and $V(P_1)$  is one,  and between $V(H)$ and $V(P_2)$ is also  one;
			\item\label{enum:tunnel_path_cut_left} $H$ is not 2-connected,   
			exactly one vertex of $P_1$ has neighbors in $H$, that is 			
			$|N_{G}(V(H))\cap V(P_1)|=1$, and no inner vertex from a  leaf-block of $H$ has a neighbor in $P_2$;
			\item\label{enum:tunnel_path_cut_right} The same as  \ref{enum:tunnel_path_cut_left}, but with $P_1$ and $P_2$ interchanged. That is, 
			$H$ is not 2-connected,  			
			$|N_{G}(V(H))\cap V(P_2)|=1$, and no inner vertex from  a leaf-block of $H$ has a neighbor in $P_1$.			
		\end{enumerate}
	\end{itemize}
	The set of \emph{\banana}s for an \bananadec 
	is defined as follows.
	First,  for each component $H$ of type \ref{enum:tunnel_path_bic}, $H$ is an \banana of the \bananadec.
	Second, for each  $H$ of type \ref{enum:tunnel_path_cut_left}, or of type \ref{enum:tunnel_path_cut_right}, all its leaf-blocks are also \banana{s} of the \bananadec.
\end{restatable}

As long as an \bananadec is available, \banana{}s allow us to bound the structure of optimal solutions in a number of ways. First, Fomin et al.~\cite{FominGSS22} observe that the longest $(s, t)$-path necessarily visits an \banana.

\begin{lemma}[Lemma 7 in~\cite{FominGSS22}]\label{lemma:st_path_edge_of_banana}
	Let $G$ be a graph and $P_1, P_2$ induce an \bananadec  for an $(s,t)$-path $P$ in $G$.
	Then there is a longest $(s,t)$-path in $G$ that enters an \banana.
\end{lemma}

Next, since an \banana has a very restrictive connection to the rest of the graph, it follows that any $(s, t)$-path has only one chance of entering the component.

\begin{lemma}[Lemma 5 in~\cite{FominGSS22}]\label{lemma:st_path_banana_consecutive}
	Let $G$ be a $2$-connected graph and $P$ be an $(s,t)$-path in $G$. Let paths  $P_1, P_2$ induce an \bananadec  for $P$ in $G$. Let  $M$ be an \banana. Then for every $(s,t)$-path $P'$ in $G$, if $P'$ enters $M$, then all vertices of $V(M)\cap V(P')$   appear consecutively in $P'$.
\end{lemma}

For the purposes of recursion, it is convenient to enclose an \banana together with some of its immediate connections,
so that this slightly larger subgraph behaves exactly like an $(s, t)$-path instance.
The subgraph $K$ in the next lemma plays this role.

\begin{lemma}[Lemma 8 in~\cite{FominGSS22}]\label{lemma:st_path_banana_to_2_connected}
	Let  paths $P_1,P_2$ induce an \bananadec  for an $(s,t)$-path $P$ in graph $G$.
	Let $M$ be an \banana in $G$. Then there is a polynomial time algorithm that outputs a  $2$-connected subgraph $K$ of $G$ and two vertices $s', t' \in V(K)$, such for that every 
	$(s,t)$-path $P'$ in $G$ that enters $M$, the following hold:
	\begin{enumerate}
		\item $V(K)=(V(M)\cup\{s',t'\})$;
		\item $P'[V(K)]$ is an $(s',t')$-subpath of $P'$ and an $(s',t')$-path in $K$;
		\item $\delta(K-\{s',t'\}) \ge \delta(G-\{s,t,s',t'\})$.
	\end{enumerate}
\end{lemma}

Most importantly, \bananadec{}s capture extremal situations, where the current $(s, t)$-path cannot be made longer in a ``simple'' way. The next lemma formalizes that intuition, stating that in polynomial time we can find either a long $(s, t)$-path or an \bananadec. The lemma is largely an analog of Lemma 4 in~\cite{FominGSS22}, however our statement here is slightly modified. Next, we recall the statement from Section~\ref{sec:overview}  and provide a proof.

\begin{restatable}{lemma}{lemmapathortunnel}\label{lemma:st_path_or_tunnel}
	Let $G$ be a   $2$-connected graph such that $\delta(G-\{s,t\})\ge 16$. 
	There is a polynomial time algorithm that 
	\begin{itemize}
		\item either outputs an  $(s,t)$-path $P$  of length at least $\min\{\frac{5}{4}\delta(G-\{s,t\})-3,|V(G)|-1\}$, 
		\item or outputs an  $(s,t)$-path $P$ with paths  $P_1, P_2$ that induce an \bananadec for $P$  in $G$.
		Additionally, there is no $(s,t)$-path in $G$ that enters at least two \bananas of this \bananadec.
	\end{itemize}
\end{restatable}

\begin{proof}
	Invoke Lemma 4 of~\cite{FominGSS22} on $G,s,t$ with $B:=\{s,t\}$ and $k:=\lfloor \delta(G-\{s,t\})/4\rfloor - 2$. Note that the condition $4k+8\le \delta(G-\{s,t\})$ required by that lemma is satisfied. Now, we either get an  $(s, t)$-path of length $\delta(G - \{s, t\}) + k$, or an $(s,t)$-path $P$ with $V(P)\cup \{s,t\}=V(G)$, or the required \bananadec with the paths $P$, $P_1$, $P_2$.
	Clearly $\delta(G-\{s,t\})+k> \frac{5}{4}\delta(G-\{s,t\})-3$, so if a path of length $\delta(G - \{s, t\}) + k$ is found, we are done.
	If an $(s,t)$-path $P$ has $V(P)\cup \{s,t\}=V(G)$, then it is a hamiltonian path in $G$, so we are done in the second case as well.
	
	If we obtain an \bananadec, then we additionally need to check whether there exists an $(s,t)$-path that goes through at least two \bananas of the \bananadec induced by $P_1$ and $P_2$ in $G$.
	To this end, iterate over all ordered pairs of \bananas in the \bananadec.
	For each pair, apply \Cref{lemma:st_path_banana_to_2_connected} to each of the two \bananas and obtain two triples $(K_1,s_1,t_1)$ and $(K_2,s_2, t_2)$.
	There is an $(s,t)$-path entering both \bananas in the order given by the pair if and only if there exist three disjoint paths between the pairs $(s,a_1),(b_1,a_2),(b_2,t)$, where $(a_i,b_i)$ is a permutation of $(s_i,t_i)$ for each $i \in [2]$.
	
	When the permutations are fixed, such paths, if they exist, can be found in polynomial time using the famous algorithm of Robertson and Seymour for \textsc{$k$-Disjoint Paths}~\cite{RobertsonS95b}.
	Since $\delta(K_i-\{s_i,t_i\})\ge \delta(G-\{s,t\})-2$ for each $i \in [2]$, these three paths together with two $(s_i,t_i)$-paths inside $K_i$ combine into an $(s,t)$-path of length at least $2\delta(G-\{s,t\})-4> \frac{5}{4}\delta(G-\{s,t\})-3$, so the algorithm outputs this path and stops.
	If the disjoint path triple was not found on any of the steps, then there is indeed no $(s,t)$-path entering at least two \bananas.
\end{proof}

Finally, to deal with $(s, t)$-paths that do not enter any \banana, one can observe the following. Intuitively, such a path should be far from optimal, as going through an \banana would immediately give at least $\delta(G-\{s,t\}) - \Oh(1)$ additional edges of the path.
The final lemma of this subsection establishes how precisely the length of a path avoiding \bananas can be ``boosted'' in this fashion.
To obtain this result, we first need a technical lemma from~\cite{FominGSS22} that yields long paths inside separable components.

\begin{lemma}[Lemma 6 in~\cite{FominGSS22}]\label{lemma:separator_in_non_2c}
	Let $H$ be a connected graph with at least one cut-vertex.
	Let $I$ be the set of inner vertices of all leaf-blocks of $H$.
	Let $S \subseteq V(H)\setminus I$ separate at least one vertex in $V(H)\setminus I$ from $I$ in $H$.
	For any vertex $v$ that is not an inner vertex of a leaf-block of $H$, there is a cut-vertex $c$ of a leaf-block of $H$ and a $(c,v)$-path of length at least $\frac{1}{2}\left(\delta(H)-|S|\right)$ in $H$. This path can be constructed in polynomial time.
\end{lemma}

Now we move to $(s, t)$-paths that avoid \bananas. The following \Cref{lemma:st_path_to_banana_st_path} has been already stated in Section~\ref{sec:overview}, here we recall the statement  and provide a proof.

\begin{restatable}{lemma}{lemtobananastpath}\label{lemma:st_path_to_banana_st_path}
	Let $P$ be an $(s,t)$-path of length at most $\delta(G-\{s,t\})+k$ and let two paths $P_1, P_2$ induce a \bananadec for $P$ in $G$.
	There is a polynomial time algorithm that, given an $(s,t)$-path of length at least $4k+5$ in $G$ that does not enter any \banana, outputs a path of length at least $\min\{\delta(G-\{s,t\})+k-1,\frac{3}{2}\delta(G-\{s,t\})-\frac{5}{2}k-1\}$ in $G$.
\end{restatable}

\begin{proof}
	For clarity, we denote $\delta:=\delta(G-\{s,t\})$.
	Let $Q$ be the given $(s,t)$-path in $G$.
	Denote by $S$ the set of the first $k$ vertices on $Q$ and by $T$ the set of the last $k$ vertices on $Q$.
	Let $s'$ be the first vertex on $Q$ that is not in $S$ and $t'$ be the last vertex on $Q$ that is not in $T$.
	Since $Q$ consists of more than $2k$ vertices, $s',t'\notin S\cup T$.
	The length of the $(s,s')$-subpath of $Q$ and the length of the $(t',t)$-subpath of $Q$ are equal to $k$.
	
	The total length of $P_1$ and $P_2$ is at most $k$, hence $|V(P_1)\cup V(P_2)|\le k+2$.
	The length of the $(s',t')$-subpath of $Q$ is at least $2k+5>2|V(P_1)\cup V(P_2)|$.
	Hence, this subpath contains at least one edge of $G$ that is not incident to vertices in $|V(P_1)\cup V(P_2)|$.
	Denote the endpoints of this edge by $u$ and $v$.
	Since $Q$ does not enter any \banana, this edge is an edge of a non-leaf-block of some separable connected component $H$ of $G-V(P_1\cup P_2)$.
	The component $H$ corresponds to either \ref{enum:tunnel_path_cut_left} or \ref{enum:tunnel_path_cut_right} in the definition of \bananadec.
	Without loss of generality, we assume that $H$ corresponds to \ref{enum:tunnel_path_cut_left}.
	
	We now consider two cases depending on the structure of $H-(S\cup T)$.
	If $S\cup T$ separates $u$ or $v$ from all cut vertices of the leaf-blocks in $H$,
	then we have a set of size $2k$ in $H$ that satisfies the condition of Lemma~\ref{lemma:separator_in_non_2c}.
	Take a vertex $w$ in $H$ that has a neighbour in $V(P_2)$ in $G$.
	By \Cref{lemma:separator_in_non_2c}, a $(w,c)$-path of length at least \[\frac{1}{2}\delta(H)-2k\ge \frac{1}{2}\delta(G-V(P_1\cup P_2))-2k\ge \frac{1}{2}\delta(G-\{s,t\})-\frac{5}{2}k-1\]
	exists in $H$ for some cut vertex $c$ of some leaf-block $L$ of $H$.
	In this leaf-block, we have a vertex $z$ with a neighbour in $V(P_1)$.
	By Corollary~\ref{thm:relaxed_st_path}, we have a $(c,z)$-path of length at least $\delta(L-c)\ge \delta(G-\{s,t\})-2$ inside $L$.
	Combine the two paths and obtain a $(z,w)$-path of length at least $\frac{3}{2}\delta(G-\{s,t\})-\frac{5}{2}k-3$ inside $H$.
	Finally, prepend to this path a prefix of $P_1$ connecting $s$ with the neighbour of $z$, and append to this path a suffix of $P_2$ connecting the neighbour of $w$ with $t$.
	The length increases by at least two as $s\neq z$ and $w\neq t$.
	The obtained path is an $(s,t)$-path of length at least 	$\frac{3}{2}\delta(G-\{s,t\})-\frac{5}{2}k-1.$
	
	The second case is when from $v$ we can reach a cut vertex $c$ of some leaf-block $L$ in $H$ while avoiding vertices in $S\cup T$.
	Note that $V(Q)\cap V(L-c)= \emptyset$ by the properties of \bananadec.
	Then choose $w$ as an arbitrary vertex in $L-c$ with a neighbour in $V(P_1)$.
	Now construct a $(v,s)$-path $Q'$ in the following way.
	First, follow an arbitrary $(v,c)$-path in $H-(S\cup T)$.
	Then continue with a $(c,w)$-path of length at least $\delta(L-c)\ge \delta(G-\{s,t\})-2$ inside $L-c$ that exists by Corollary~\ref{thm:relaxed_st_path}.
	Note that this path has no common vertices with $Q$.
	Finish $Q'$ by going from $w$ to the neighbour of $w$ in $P_1$ and follow $P_1$ backwards down to $s$.
	
	Let $x$ be the last vertex before $c$ on $Q'$ that belongs to $V(Q)$.
	Let $y$ be the first vertex after $w$ on $Q'$ that belongs to $V(Q)$.
	Both $x,y$ are defined correctly since $v,s\in V(Q)$.
	Consider the $(x,y)$-subpath of $Q'$.
	It strictly contains the $(c,w)$-path inside $L$, so its length is at least $\delta(G-\{s,t\})-1$.
	Also, the length of the $(s,x)$-subpath of $Q$ and the length of the $(x,t)$-subpath of $Q$ is at least $k$ as $x\notin S\cup T$.
	
	We now construct a long $(s,t)$-path in $G$.
	If $y$ is contained in the $(s,x)$-subpath of $Q$, then the $(s,t)$-path is constructed in the following way: follow $P_1$ from $s$ to $y$, then follow $Q'$ backwards from $y$ down to $x$, and finish by following $Q$ from $x$ to $t$.
	The length of this path is at least $\delta(G-\{s,t\})-1+k$.
	If $y$ belongs to the $(x,t)$-subpath of $Q$, start by taking the $(s,x)$-subpath of $Q$, then follow $Q'$ from $x$ to $y$ and finish by following $P_2$ from $y$ to $t$.
	This path also has length at least $k+\delta(G-\{s,t\})-1$.
	The proof is complete.
\end{proof}

\subsection{Proof of Theorem~\ref{thm:eg_approx}}

To deal with the recursive structure of the solution, we introduce the following \emph{nested} generalization of an \bananadec. Intuitively, it captures how the structural observations of the previous subsection allow us to recursively construct \bananadec{}s with the aim of finding a long $(s,t)$-path. For an illustration of a \nestedbananadec, see \Cref{fig:nested_banana_example}. We recall the formal definition from Section~\ref{sec:overview}.

\begin{restatable}[\Nestedbananadec]{definition}{defnestedbanana}\label{def:nested_banana}
	A sequence of triples $(G_1,s_1,t_1)$, $(G_2,s_2,t_2)$, \ldots, $(G_\ell,s_\ell,t_\ell)$ is called a \emph{\nestedbananadec} for $G$ and two vertices $s,t \in V(G)$ if
	
	\begin{itemize}
		\item $(G_1,s_1,t_1)=(G,s,t)$;
		\item for each $i \in [\ell]$, either
		\begin{itemize}
			\item $\delta(G_i-\{s_i,t_i\})<16$, or
			\item \Cref{lemma:st_path_or_tunnel} applied to $G_i,s_i,t_i$ gives a path $P_i$ of length at least $\min\{\frac{5}{4}\delta(G_i-\{s_i,t_i\})-3,|V(G_i)|-1\}$ in $G_i$, or
			\item \Cref{lemma:st_path_or_tunnel} applied to $G_i,s_i,t_i$ gives a path $P_i$ and two paths $P_{i,1},P_{i,2}$ that induce an \bananadec for $P_i$ in $G_i$, and
			for each \banana $M$ of this decomposition there is $j>i$ such that $(G_j,s_j,t_j)$ is the result of \Cref{lemma:st_path_banana_to_2_connected} applied to $M$ in $G_i$.
			In this case, we say that $G_i$ is \emph{decomposed}.
		\end{itemize}
		\item for each $i \in \{2,\ldots,\ell\}$, there is $e(i)<i$ such that $(G_i,s_i,t_i)$ is a result of \Cref{lemma:st_path_banana_to_2_connected} applied to some \banana of the \bananadec of $G_{e(i)}$ for $P_{e(i)}$.
	\end{itemize}
\end{restatable}

The proof of Theorem~\ref{thm:eg_approx} is performed in two steps: first, we show how to obtain a \nestedbananadec for a given graph $G$, and then we use the \nestedbananadec to recursively construct a good approximation to the longest $(s, t)$-path. The first part is achieved simply by applying Lemma~\ref{lemma:st_path_or_tunnel} recursively on each \banana until components are no longer decomposable. The main hurdle is the second part, on which we focus for the rest of the section. For completeness, first we show that a \nestedbananadec can always be constructed in polynomial time.

\begin{lemma}
    \label{lemma:nested_construction}
    There is a polynomial time algorithm that, given a 2-connected graph $G$ and its two vertices $s$ and $t$, outputs a \nestedbananadec for $G$, $s$, $t$.
\end{lemma}

\begin{proof}
    The algorithm proceeds recursively, starting with the triple $(G_1, s_1, t_1) = (G, s, t)$. For the given triple $(G_i, s_i, t_i)$,
    if $\delta(G_i-\{s_i,t_i\})<16$, the algorithm stops. Otherwise, invoke the algorithm of \Cref{lemma:st_path_or_tunnel} on $(G_i, s_i, t_i)$. If this returns a path $P_i$ of length at least $\frac{5}{4}\delta(G_i-\{s_i,t_i\})-3$, the algorithm stops. On the other hand, if an \bananadec is returned, for each \banana $M$ run the algorithm of Lemma~\Cref{lemma:st_path_banana_to_2_connected} on $M$ to obtain a triple $(G_j, s_j, t_j)$, where $j$ is the lowest free index among the triples produced so far. Run the main algorithm recursively on each of the triples generated on this step.

    By definition, the algorithm above produces a \nestedbananadec. To show that the running time is polynomial,
    first observe that running the algorithm without the subsequent recursive calls is clearly polynomial. Assume this running time is bounded by $\alpha n^c$ for some constant $\alpha > 0$ and $c \ge 1$, where $n = |V(G) \setminus \{s, t\}|$ and ($G$, $s$, $t$) is the current instance. We show by induction on the depth of the resulting \nestedbananadec that the running time of the recursive algorithm is at most $\alpha n^{c + 1}$. If the instance does not spawn any recursive calls, this trivially holds. Otherwise, assume $\ell$ new instances $(G_{j_1}, s_{j_1}, t_{j_1})$, \ldots, $(G_{j_\ell}, s_{j_\ell}, t_{j_\ell})$ are produced, denote $n_i = |V(G_{j_i} \setminus \{s_{j_i}, t_{j_i}\}|$. Note that $\ell \ge 2$ since there are always at least 2 components in an \bananadec.
    By induction, the running time is bounded by $\alpha \cdot \left(n^c + \sum_{i = 1}^\ell n_i^{c + 1}\right)$. We now bound $\sum_{i = 1}^\ell n_i^{c + 1}$, observe first that $\sum_{i = 1}^\ell n_i \le n$, as all the sets $V(G_{j_i} \setminus \{s_{j_i}, t_{j_i}\}$ are disjoint and do not contain $s$ or $t$. We use the following numerical observation proven in~\cite{FominGSS22}.
    \begin{claim}[Proposition 3 in~\cite{FominGSS22}]\label{claim:sum_powers}
		Let $a_1, a_2, \ldots, a_q$ be a sequence of $q\ge 2$ positive integers with $\sum_{i=1}^q a_i=n$.
		Let $x>1$ be an integer.
		Then $\sum_{i=1}^q a_i^x \le (n-1)^x+1 \le n^x - n^{x-1}$.
    \end{claim}
    By Claim~\ref{claim:sum_powers}, we can bound the running time by
    \[\alpha \cdot \left(n^c + \sum_{i = 1}^\ell n_i^{c + 1}\right) \le \alpha \cdot \left(n^c + n^{c + 1} - n^c \right) = \alpha n^{c + 1},\]
    completing the proof.

\end{proof}

Clearly, it follows that the size of a \nestedbananadec returned by Lemma~\ref{lemma:nested_construction} is also polynomial.
Observe also that the construction algorithm invokes \Cref{lemma:st_path_or_tunnel} for all sufficiently large $G_i$, thus in what follows we assume that the corresponding paths $P_i$ are already computed.

Now we focus on using a constructed \nestedbananadec for approximating the longest $(s, t)$-path. First of all, we present the algorithm \texttt{\LongNestedPathName} that, given a \nestedbananadec of $G$, computes a long $(s, t)$-path by going over the decomposition. The pseudocode of \texttt{\LongNestedPathName} is present in \Cref{alg:long_nested_path}.
Intuitively, first the algorithm computes a compression $H$ of the graph $G$ that respects the \nestedbananadec: components that are not decomposed are replaced by single edges, and edges that are ``unavoidable'' to visit a component are contracted. The computation of this compression is encapsulated in the \texttt{\NestedCompressionName} function presented in \Cref{alg:nested_compress}.
As a subroutine, this function uses the \texttt{two\_long\_disjoint\_paths} algorithm given by \Cref{cor:two_long_st_paths}, that finds two disjoint paths of at least the given length between the given pairs of vertices.

Next, the blackbox approximation algorithm \texttt{long\_st\_path\_approx} is used to compute an $(s, t)$-path $Q$ in $H$. The function \texttt{\NestedDecompressionName} reconstructs then this path in the original graph $G$, see \Cref{alg:nested_decompress} for the pseudocode. Later we argue (\Cref{lemma:path_in_H}) that any $(s,t)$-path in $H$ of length $r$ yields in this way an $(s, t)$-path in $G$ of length at least $\delta(G-\{s,t\})+r/8-3$. Finally, either the length of $Q$ in $H$ was large enough and the reconstructed path provides the desired approximation or a long path can be found inside one of the components in a ``simple'' way, and then connected arbitrarily to $\{s, t\}$. Specifically, in this component, it suffices to either take an approximation of the longest path computed by \texttt{long\_st\_path\_approx}, or a long Erd{\H{o}}s--Gallai path returned by the algorithm from \Cref{thm:relaxed_st_path}, \texttt{long\_eg\_st\_path}.
Thus, in the final few lines \texttt{\LongNestedPathName} checks whether any of these paths is longer than the reconstructed path $Q$.
The path from inside the component is extended to an $\{s, t\}$-path in $G$ by using the algorithm \texttt{two\_long\_disjoint\_paths}, given by \Cref{cor:two_long_st_paths}, with the parameter $0$.

\IncMargin{1em} 
\begin{algorithm}[h]
	\SetKwFunction{LongestPath}{longest\_path}
	\SetKwFunction{LongCycleApprox}{longest\_cycle\_approx}
	\SetKwFunction{LongSTPath}{long\_st\_path}
	\SetKwFunction{LongErdosSTPathApprox}{long\_eg\_st\_path\_approx}
	\SetKwFunction{LongSTPathApprox}{long\_st\_path\_approx}
	\SetKwFunction{LongErdosSTPath}{long\_eg\_st\_path}
	\SetKwFunction{LongSTCycle}{long\_st\_cycle}
	\SetKwFunction{LongDiracCycle}{long\_dirac\_cycle}
	\SetKwFunction{HamPath}{hamiltonian\_path}
	\SetKwFunction{LongNestedSTPath}{\LongNestedPathName}
	\SetKwFunction{NestedCompression}{\NestedCompressionName}
	\SetKwFunction{TwoLongDisjointPaths}{two\_long\_disjoint\_paths}
	\let\oldnl\nl
	\newcommand{\nonl}{\renewcommand{\nl}{\let\nl\oldnl}}
	\Indm \nonl \NestedCompression{$(G_1,s_1,t_1),(G_2,s_2,t_2),\ldots,(G_\ell,s_\ell,t_\ell)$}
	
	\Indp
	
	\KwIn{a \nestedbananadec for $G$, $s$ and $t$.}
	\KwOut{the compressed graph $H$.}

	$H \longleftarrow G$\;
	\ForEach{$i \in \{2,\ldots,\ell\}$}{
		$j \longleftarrow e(i)$\;
		$d_i \longleftarrow |\{s_j,t_j\}\setminus\{s_{i},t_{i}\}|$\;
		\If{\TwoLongDisjointPaths$(G_i, \{s_j,t_j\},\{s_i,t_i\},d_i+1)$ is \textsc{No}}{
			contract all edges of a maximum matching between $\{s_j,t_j\}$ and $\{s_i,t_i\}$ in $H$\label{line:edge_contraction}\;
		}
		\If{$G_i$ is not decomposed}{
			remove all vertices in $V(G_i)\setminus \{s_i,t_i\}$ from $H$\;
			add edge $s_it_i$ to $H$ and mark it with $G_i$\;
		}
	}

	\Return $H$\;
	\caption{The algorithm compressing a given graph $G$ with a given \nestedbananadec.}\label{alg:nested_compress}
\end{algorithm}

\begin{algorithm}[ht]
	\SetKwFunction{LongestPath}{longest\_path}
	\SetKwFunction{LongCycleApprox}{longest\_cycle\_approx}
	\SetKwFunction{LongSTPath}{long\_st\_path}
	\SetKwFunction{LongErdosSTPathApprox}{long\_eg\_st\_path\_approx}
	\SetKwFunction{LongSTPathApprox}{long\_st\_path\_approx}
	\SetKwFunction{LongErdosSTPath}{long\_eg\_st\_path}
	\SetKwFunction{LongSTCycle}{long\_st\_cycle}
	\SetKwFunction{LongDiracCycle}{long\_dirac\_cycle}
	\SetKwFunction{HamPath}{hamiltonian\_path}
	\SetKwFunction{LongNestedSTPath}{\LongNestedPathName}
	\SetKwFunction{TwoLongDisjointPaths}{two\_long\_disjoint\_paths}
	\SetKwFunction{NestedDecompression}{\NestedDecompressionName}
	\let\oldnl\nl
	\newcommand{\nonl}{\renewcommand{\nl}{\let\nl\oldnl}}
	\Indm \nonl \NestedDecompression{$(G_1,s_1,t_1),(G_2,s_2,t_2),\ldots,(G_\ell,s_\ell,t_\ell),H,Q$}
	
	\Indp
	
	\KwIn{a \nestedbananadec for $G, s$ and $t$, the compressed graph $H$ and an $(s,t)$-path $Q$ in $H$ of length $r$.}
	\KwOut{an $(s,t)$-path of length at least $\delta(G-\{s,t\})+r/8-3$ in $G$.}

	\ForEach{$i \in \{2,\ldots,\ell\}$ such that $d_i>0$ and $Q$ enters $G_i$}{\label{line:Q_transformation_start}
		$j \longleftarrow e(i)$\;
		\eIf{an edge between $\{s_j,t_j\}$ and $\{s_i,t_i\}$ was contracted in $H$}{
			replace $s_i$ and/or $t_i$ in $Q$ with the respective contracted edges;}
		{
			$S_1, S_2 \longleftarrow \TwoLongDisjointPaths(G, \{s_j,t_j\},\{s_i,t_i\},d_i+1)$\;
			replace the two subpaths of $Q$ going from $\{s_j,t_j\}$ to $\{s_i,t_{i}\}$ with $S_1$ and $S_2$ if the length of $Q$ increases\label{line:two_paths_replacement}\;
		}
	}\label{line:Q_replacement_end}

	$h \longleftarrow $ largest $h \in [\ell]$ such that $Q$ enters $G_h$\;
	\eIf{$G_h$ is not decomposed}{
		replace $s_ht_h$ in $Q$ with $P_h$\;\label{line:Q_replace_marked}
	}{
        $k' \longleftarrow \lfloor(|E(Q) \cap E(G_h)| - 5) / 8\rfloor$\;
        \eIf{$|E(P_h)| \ge \delta(G_h - \{s_h, t_h\}) + k'$}{
	        $R \longleftarrow P_h$\;
        } {
		    $R \longleftarrow $ result of \Cref{lemma:st_path_to_banana_st_path} applied to $G_h,P_h$ and the $(s_h,t_h)$-subpath of $Q$\;
	    }
		\If{$(s_h,t_h)$-subpath of $Q$ is shorter than $R$}{
			replace the $(s_h,t_h)$-subpath of $Q$ with $R$\;
		}
	}\label{line:Q_transformation_end}
	
	\Return $Q$\;
	\caption{The algorithm decompressing a path in $H$ into a long path in $G$.}\label{alg:nested_decompress}
\end{algorithm}

\begin{algorithm}[ht]
	\SetKwFunction{LongestPath}{longest\_path}
	\SetKwFunction{LongCycleApprox}{longest\_cycle\_approx}
	\SetKwFunction{LongSTPath}{long\_st\_path}
	\SetKwFunction{LongErdosSTPathApprox}{long\_eg\_st\_path\_approx}
	\SetKwFunction{LongSTPathApprox}{long\_st\_path\_approx}
	\SetKwFunction{LongErdosSTPath}{long\_eg\_st\_path}
	\SetKwFunction{LongSTCycle}{long\_st\_cycle}
	\SetKwFunction{LongDiracCycle}{long\_dirac\_cycle}
	\SetKwFunction{HamPath}{hamiltonian\_path}
	\SetKwFunction{LongNestedSTPath}{\LongNestedPathName}
	\SetKwFunction{TwoLongDisjointPaths}{two\_long\_disjoint\_paths}
	\let\oldnl\nl
	\newcommand{\nonl}{\renewcommand{\nl}{\let\nl\oldnl}}
	\Indm \nonl \LongNestedSTPath{$(G_1,s_1,t_1),(G_2,s_2,t_2),\ldots,(G_\ell,s_\ell,t_\ell)$}
	
	\Indp
	
	\KwIn{a \nestedbananadec for $G, s$ and $t$.}
	\KwOut{an $(s,t)$-path of length at least $\delta(G-\{s,t\})+f(k)/32-3$ in $G$ where $k=\pathlength-\delta(G-\{s,t\})$ for the longest $(s,t)$-path length \mathpathlength in $G$.}

	$H \longleftarrow \NestedCompression((G_1,s_1,t_1),(G_2,s_2,t_2),\ldots,(G_\ell,s_\ell,t_\ell))$\;
	
	$Q \longleftarrow \LongSTPathApprox(H,s,t)$\label{line:nested_H_finished}\label{line:Q_approx}\;	
	$Q\longleftarrow \NestedDecompression((G_1,s_1,t_1),(G_2,s_2,t_2),\ldots,(G_\ell,s_\ell,t_\ell),H,Q)$\;
	\ForEach{$i \in [\ell]$}{
		$P_i \longleftarrow $ the longest of $\{\LongSTPathApprox(G_i,s_i,t_i),\LongErdosSTPath(G_i,s_i,t_i)\}$\label{line:longest_P_h}\;
		$Q \longleftarrow $ the longest of $\{Q,\TwoLongDisjointPaths(G,\{s,t\},\{s_i,t_i\}, 0)\cup P_i\}$\label{line:outer_appendage}\;
	}
	
	\Return $Q$\;
	\caption{The algorithm finding a long $(s,t)$-path in a $2$-connected graph with a given \nestedbananadec.}\label{alg:long_nested_path}
\end{algorithm}

\DecMargin{1em}

Now, our goal is to show that the path that the \texttt{\LongNestedPathName} algorithm constructs serves indeed as the desired approximation of the longest $(s, t)$-path in $G$. For the rest of this section, let $G_1, \ldots, G_\ell$ be the given \nestedbananadec for $G, s,t$. An important piece of intuition about \nestedbananadec is that, as we go deeper into the nested \bananas, the minimum degree of the component $\delta(G_i \setminus \{s_i, t_i\})$ decreases, but we gain more and more edges that we collect while going from $\{s, t\}$ to $\{s_i, t_i\}$. We introduce values that help us measure this difference between the nested components: for each $i\in[\ell]$, denote $d_i=|\{s_{e(i)},t_{e(i)}\}\setminus \{s_i,t_i\}|$.
In particular, by \Cref{lemma:st_path_banana_to_2_connected} we know that for any $i \in [\ell]$, $\delta(G_{i})\ge \delta(G_{e(i)})-d_{i}$. On the other hand, any pair of disjoint paths that connects $\{s_{e(i)}, t_{e(i)}\}$ to $\{s_i, t_i\}$ contains at least $d_i$ edges.
This leads to the following simple observation about extending an $(s_j, t_j)$-path in a component $G_j$ to an $(s, t)$-path in $G$.

\begin{claim}\label{claim:at_least_sumx}
	For each $j\in [\ell]$, let $G_{j_1},\ldots, G_{j_c}$ be such that $j_c=j$ and $j_1=1$ and $e(j_{i+1})=j_i$ for each $i\in[c-1]$.
	Let $P$ be an $(s_j, t_j)$-path in $G_j$.
    Then $P$ combined with any pair of disjoint paths connecting $\{s, t\}$ to $\{s_j, t_j\}$ yields an $(s, t)$-path in $G$ of length at least $|E(P)|+\sum_{i\in [c-1]} d_{j_{i+1}}$.
\end{claim}

However, there might also exist longer paths connecting nested components $G_{e(i)}$ and $G_i$.
When we construct the compressed graph $H$ in \Cref{alg:nested_compress}, we distinguish between two cases. Either any pair of such paths have the  total length $d_i$, meaning that the only option is to use the edges of a matching between $\{s_{e(i)}, t_{e(i)}\}$ and $\{s_i, t_i\}$. In that case we simply contract these edges as we know that there is no choice on how to reach $G_i$ from $G_{e(i)}$. Or, there is a pair of disjoint paths of total length at least $d_i + 1$. This situation is beneficial to us in a different way: since we can find such a pair of paths in polynomial time, we can traverse at least $d_i + 1$ edges going from $G_{e(i)}$ to $G_i$, while we only lose at most $d_i$ in the minimum degree. This dichotomy on the structure of the ``slice'' between two nested components is the main leverage that allows us to lift the length of an $(s,t)$-path in $H$ to an offset above the minimum degree in $G$.
We formally show this crucial property of the compressed graph $H$ and the \texttt{\NestedDecompressionName} routine in the next lemma.

	\begin{lemma}\label{lemma:path_in_H}
		The \texttt{\NestedDecompressionName} routine transforms an $(s,t)$-path $Q$ in $H$ of length $r$ into an $(s,t)$-path in $G$ of length at least $\delta(G-\{s,t\})+r/8-3$.
	\end{lemma}
	
	\begin{proof}
        Observe that in the tree of the \nestedbananadec, the path $Q$ visits a rooted subpath of components $G_i$.
		That is, there are indices $j_1,j_2,\ldots, j_c \in [\ell]$ such that $j_1=1$ and $e(j_{i+1})=j_i$ for each $i\in[c-1]$.
		This holds since in a \bananadec on each level, $Q$ visits at most one \banana by \Cref{lemma:st_path_or_tunnel}.
		Here we say that $Q$ visits a component $G_i$ if $Q$ contains an edge of $G_i$ that was not contracted in $H$, and for non-decomposed components $G_i$ this means that $Q$ contains the edge $s_it_i$ in $H$.

		By \Cref{lemma:st_path_banana_to_2_connected}, $\delta(G_{j_{i+1}})\ge \delta(G_{j_i})-d_{j_{i+1}}$.
		Let $h \in [\ell]$ be the largest integer such that $Q$ enters $G_h$, $h = j_c$. Denote by $p$ be the number of edges in $E(Q)\setminus E(G_h)$ and by $y$ the length of the $(s_h,t_h)$-subpath of $Q$, then $p+q=r$.

		We now analyze the length of $Q$ after performing the replacement operations in Lines~\ref{line:Q_transformation_start}--\ref{line:Q_replacement_end}.
		Denote by $Y$ the set of all $i \in [c-1]$ such that no contraction was made in Line~\ref{line:edge_contraction} between $\{s_{j_i},t_{j_i}\}$ and $\{s_{j_{i+1}},t_{j_{i+1}}\}$.
		For each $i\in Y$ with $d_{j_{i+1}}>0$, performing the replacement operation in Line~\ref{line:two_paths_replacement} between $\{s_{j_i},t_{j_i}\}$ and $\{s_{j_{i+1}},t_{j_{i+1}}\}$ in $Q$ yields $$|E(Q)\cap E(G_{j_i})\setminus E(G_{j_{i+1}})|\ge d_{j_{i+1}}+1\ge 3d_{j_{i+1}}/2.$$
		Let $p'$ be the length of $Q$ outside of $G_h$ after all these replacements, from the above $p' \ge \frac{3}{2}\sum_{i\in Y}d_{j_{i+1}}$. Also, $p' \ge p$ since the replacement only takes place if it makes the path longer.

		Denote by $X:=[c-1]\setminus Y$ the set of all $i \in [c-1]$ such that a contraction was made in Line~\ref{line:edge_contraction} between $\{s_{j_i},t_{j_i}\}$ and $\{s_{j_{i+1}},t_{j_{i+1}}\}$.
		For each $i\in X$ the algorithm reverses the respective edge contractions done in Line~\ref{line:edge_contraction} in $Q$.
		This increases the length of $Q$ by $d_{j_{i+1}}$, so after Line~\ref{line:Q_replacement_end} it holds that
		$|E(Q) \setminus E(G_h)| \ge p' +\sum_{i\in X}d_{j_{i+1}}.$

        We now observe that any long $(s_h, t_h)$-subpath in $G_h$ can be combined with $Q$ to preserve at least a constant fraction of $p$ in the offset.
		\begin{claim}\label{claim:outer_offset}
		    After Line~\ref{line:Q_replacement_end}, replacing the $(s_h, t_h)$-subpath of $Q$ with a path $P$ in $G_h$ of length $\delta(G_h-\{s_h,t_h\}) + k'$, where $k'$ is a nonnegative integer, yields an $(s, t)$-path in $G$ of length at least \[\delta(G-\{s,t\}) + k' + p/3.\]
		\end{claim}
		\begin{claimproof}
		The length of the resulting path is at least

		\begin{multline*}
		|E(Q)\setminus E(G_{j})|+|E(P)| \ge p'+\sum_{i\in X}d_{j_{i+1}}+\delta(G_h-\{s_h,t_h\}) + k'\\ \ge p'+\sum_{i\in X}d_{j_{i+1}}+\delta(G-\{s,t\})-\sum_{i\in[c-1]}d_{j_{i+1}} + k'\ge
		\delta(G-\{s,t\})+k' + p'-\sum_{i\in Y}d_{j_{i+1}} \\ \ge\delta(G-\{s,t\}) + k' + p/3.
		\end{multline*}

		Note that the last inequality holds since $p'$ is at least $\frac{3}{2}\sum_{i\in Y}d_{j_{i+1}}$ and also at least $p$.
		The path obtained at this point is an $(s,t)$-path in $G$ with possibly some contracted edges, since not all edge contractions were reversed. Reverse all remaining edge contractions affecting $Q$ and obtain an $(s,t)$-path in $G$ of at least the same length.
		    
		\end{claimproof}

		For estimating the length of the $(s_h,t_h)$-subpath, consider two cases depending on the type of $G_h$.
		
		\textbf{$G_h$ is not decomposed.}
        In this case, $Q$ contains the edge $s_ht_h$ in $H$, and in Line~\ref{line:Q_replace_marked} this edge is replaced with the path $P_h$.
        By \Cref{claim:outer_offset}, this yields a path of length at least $\delta(G-\{s,t\}) + (r - 1)/3$, since the length 
        of $P_h$ is at least $\delta(G_h-\{s_h,t_h\})$, and $p = r - 1$.

		\textbf{$G_h$ is decomposed.}
		By the choice of $j$, the $(s_h,t_h)$-subpath of $Q$ does not enter any \banana in the \bananadec induced by $P_{j,1}$ and $P_{j,2}$ in $G_h$.
				
		By \Cref{claim:outer_offset}, an $(s_h,t_h)$-path of length $\delta(G_h-\{s_h,t_h\})+k'$ inside $G_h$ combined with the outer part of $Q$ obtains an $(s,t)$-path of length at least $\delta(G-\{s,t\})+p/3+k'$ inside $G$.
		We now focus on identifying a long enough $(s_h,t_h)$-path inside $G_h$.
		
		Let $k':=\lfloor (q-5)/8\rfloor$, so $q \ge 8k'+5\ge 4k'+5$.
		If $P_h$ is longer than $\delta(G_h-\{s_h,t_h\})+k'$, then plugging $P_h$ into \Cref{claim:outer_offset} gives an $(s,t)$-path of length at least $\delta(G-\{s,t\})+p/3+k'+1\ge\delta(G-\{s,t\})+p/3+(q-5)/8>\delta(G-\{s,t\})+r/8-1$.
		Otherwise, we apply \Cref{lemma:st_path_to_banana_st_path} to $G_h$, $P_h$ and the $(s_h,t_h)$-subpath of $Q$ to obtain an $(s_h, t_h)$-path $R$ in $G_h$.

		If the length of $R$ is at least $\delta(G_h-\{s_h,t_h\})+k'-1$, \Cref{claim:outer_offset} gives the desired bound of $\delta(G-\{s,t\})+r/8-3$.
		Otherwise, $\frac{1}{2}\delta(G_h-\{s_h,t_h\})-\frac{5}{2}k'<k'$, then $7k'>\delta(G_h-\{s_h,t_h\})$.
		It follows that $q > \delta(G_h-\{s_h,t_h\}) + q/8+5$.
		Hence, by applying \Cref{claim:outer_offset} to the initial $(s_h,t_h)$-subpath of $Q$ we get a path of length at least $\delta(G-\{s,t\})+p/3+q/8+5> \delta(G-\{s,t\})+r/8$.
		Since \Cref{alg:nested_decompress} takes the longest of $R$ and the original subpath of $Q$, both cases are covered.
	\end{proof}

    It will also be helpful to observe that in the ``slice'' between a decomposed component and the nested components, at most two edges of any path can be contracted. Note that this does not follow immediately, as a pair of edges to \emph{each} of the nested components is potentially contracted.

	\begin{claim}\label{claim:two_edges_affected}
		Let $Q$ be an $(s_j,t_j)$-path inside a decomposed graph $G_j$.
		Then all edges $E(Q)\cap E(G_j)\setminus \bigcup_{e(i)=j}E(G_i)$ are unchanged in $H$ except for, possibly, contraction of the first and the last edge of $Q$.
	\end{claim}
	
	\begin{claimproof}
		Let $i$ be such that a contraction is made for $G_i$ in  Line~\ref{line:edge_contraction} with $e(i)=j$ and $d_i>0$.
		There are no two disjoint paths between $\{s_j,t_j\}$ and $\{s_i,t_i\}$ of total length at least $d_i+1$.
		
		Without loss of generality, we assume that $s_i\neq s_j$, $t_i\neq s_j$ and $s_i\neq t_j$ and the edge $s_js_i$ is contracted.
		If $s_i \notin V(Q)$, then $Q$ is not affected in $H$.
		We assume that $s_i \in V(Q)$.
		If $s_i$ is the second vertex in $V(Q)$, then $s_js_i$ is the first edge of $V(Q)$ as required.
		
		Suppose now that $s_i$ is not the second vertex in $Q$.
		Then the $(s_j,s_i)$-subpath of $Q$ is of length at least two.
		If $t_i$ does not belong to this subpath, we add the trivial (of length zero or one) $(t_i,t_j)$-path in $G_j$ and obtain two disjoint paths between $\{s_j,t_j\}$ and $\{s_i,t_i\}$ of total length at least $2+|\{t_i,t_j\}|-1>d_i$, which is a contradiction.
		
		Hence, $t_i$ is present on the $(s_j,s_i)$-subpath of $Q$.
		Then $t_i\neq t_j$, so $d_i=2$ and $s_i,t_i,s_j,t_j$ are all distinct.
		We have an $(s_j,t_i)$-subpath of $Q$ and an $(s_i,t_j)$-subpath of $Q$ which are disjoint.
		If one of them is of length at least two, then we have two disjoint paths of total length more than $d_i$.
		Hence, $s_jt_i$ and $s_it_j$ are the first and the last edge in $Q$.
		The proof of the claim is complete.
		
	\end{claimproof}

    Now we are ready to prove the main lemma that bounds the length of the $(s, t)$-path returned by \Cref{alg:long_nested_path}.

\begin{lemma}
	\texttt{\LongNestedPathName} outputs an $(s,t)$-path in $G$ of length at least $\delta(G-\{s,t\})+f(k)/32-3$, where $k=L-\delta(G-\{s,t\})$ and $L$ is the length of the longest $(s,t)$-path in $G$.
\end{lemma}

\begin{proof}
	Let $T$ be the longest $(s,t)$-path in $G$, $|E(T)|=L=\delta(G-\{s,t\})+k$.
	Our aim is to show that either $T$ yields a sufficiently long path in $H$ to use \Cref{lemma:path_in_H}, or conclude that
	after contractions most of the path stays inside one \banana, the deepest component visited.
	In case of the latter, we show that it suffices to take a long path inside this component.

	We now introduce some notations for $T$ with respect to the \nestedbananadec structure, similarly to the proof of \Cref{lemma:path_in_H}.
	Let $h\in [\ell]$ be the largest integer such that $T$ enters $G_h$.
	Let $j_1,j_2,\ldots, j_c$ be such that $j_c=j$, $j_1=1$ and $e(j_{i+1})=j_i$ for each $i\in[c-1]$.
	Denote by $Y$ the set of all $i \in [c-1]$ such that no contraction was made in Line~\ref{line:edge_contraction} between $\{s_{j_i},t_{j_i}\}$ and $\{s_{j_{i+1}},t_{j_{i+1}}\}$.
	Denote $X:=[c-1]\setminus Y$.
	
	Consider what happens to the path $T$ in the graph $H$ between two consecutive nested components $G_{j_i}$ and $G_{j_{i + 1}}$.
	Since $T$ is the longest $(s,t)$-path in $G$, edges in $E(T)\cap E(G_{j_i})\setminus E(G_{j_{i+1}})$ form two disjoint paths between $\{s_{j_i},t_{j_i}\}$ and $\{s_{j_{i+1}},t_{j_{i+1}}\}$ of longest possible total length.
	If $|E(T)\cap E(G_{j_i})\setminus E(G_{j_{i+1}})|=d_{j_{i+1}}$, all edges in $E(T)\cap E(G_{j_i})\setminus E(G_{j_{i+1}})$ are contracted in $H$.

	Otherwise, $|E(T)\cap E(G_{j_i})\setminus E(G_{j_{i+1}})|>d_{j_{i+1}}$.
	By \Cref{claim:two_edges_affected}, all edges in $E(T)\cap E(G_{j_i})\setminus E(G_{j_{i+1}})$ are present in $H$, except for possibly $d_{j_{i+1}}$ of them (the first and/or the last).
	Also, recall that by properties of \nestedbananadec, $T$ does not enter any $G_i$ with $e(i)=j_i$ except for $G_{j_{i+1}}$.
	Then the removal of the internal vertices of non-decomposed components does not affect edges in $E(T)\cap E(G_{j_i})\setminus E(G_{j_{i+1}})$.
	Hence, at least $|E(T)\cap E(G_{j_i})\setminus E(G_{j_{i+1}})|-d_{j_{i+1}}$ of the edges are present in $H$ in this case, which is at least one third of the edges in $E(T)\cap E(G_{j_i})\setminus E(G_{j_{i+1}})$ since $d_{j_{i + 1}} \le 2$.

	Let $T'$ be the path $T$ with all contractions applied to $H$.
	If $G_h$ is not decomposed, we assume that $T'$ contains the edge $s_ht_h$ marked with $G_h$.
	By the above, we have $$|E(T')\setminus E(G_h)|\ge \sum_{i\in[c-1]}|E(T)\cap E(G_{j_i})\setminus E(G_{j_{i+1}})|-d_{j_{i+1}}\ge\frac{1}{3}(|E(T)\setminus E(G_h)|-\sum_{i\in X}d_{j_{i+1}}).$$
	The last inequality holds since for each $i \in Y$ with $d_{j_{i + 1}} > 0$, $|E(T)\cap E(G_{j_i})\setminus E(G_{j_{i+1}})| -d_{j_{i+1}}$ is at least $\frac{1}{3}|E(T)\cap E(G_{j_i})\setminus E(G_{j_{i+1}})|$, and for all the remaining indices $i$, $|E(T)\cap E(G_{j_i})\setminus E(G_{j_{i+1}})| = d_{j_{i+1}}$.
	Denote $p:=|E(T')\setminus E(G_h)|$, and from the above obtain the equivalent
	\begin{equation}
	|E(T)\setminus E(G_h)|\le 3p+\sum_{i\in X}d_{j_{i+1}}.
	    \label{eq:outside_path_bound}
	\end{equation}
	
	If $p \ge k/4$, then an $(s,t)$-path of length at least $k/4+1$ is present in $H$.
	In this case, an approximation of the longest $(s,t)$-path in $H$ in Line~\ref{line:Q_approx} of \texttt{\LongNestedPathName} gives a path of length at least $f(k/4+1)\ge f(k)/4$.
	By \Cref{lemma:path_in_H}, running \texttt{\NestedDecompressionName} on this path results in an $(s,t)$-path of length at least $\delta(G-\{s,t\})+f(k)/32-3$ in $G$, so in this case the proof is finished.

	Otherwise, $p < k/4$.
	Denote by $T_h$ the $(s_h,t_h)$-subpath of $T$.
	For simplicity, denote $\delta:=\delta(G-\{s,t\})$ and $\delta_h:=\delta(G_h-\{s_h,t_h\})$.
	We consider two cases.
	
	\textbf{$G_h$ is decomposed.}
	In this case, $T_h$ does not enter any \banana of $G_h$.
	By Claim~\ref{claim:two_edges_affected}, at most two edges of $T_h$ can be contracted in $H$, denote the number of such edges by $d_h'$. Also, at least one edge is not contracted, so $|E(T_h)| - d_h' \ge |E(T_h)|/3$. Thus, $T'$ is an $(s,t)$-path of length at least $p+|E(T_h)| - d_h'\ge\frac{1}{3}\left(|E(T)|-\sum_{i\in X}d_{j_{i+1}}\right)$ in $H$.
	If $\sum_{i\in X}d_{j_{i+1}} \ge \delta+k/4-3$, then \texttt{\LongNestedPathName} outputs an $(s,t)$-path of length at least $\delta+k/4-3$ by \Cref{claim:at_least_sumx}, regardless of the length of $P_h$ at Line~\ref{line:longest_P_h}.

	Otherwise, $T'$ is an $(s,t)$-path in $H$ of length at least $\frac{1}{3}(3k/4+3) =k/4+1$.
    Analogously to the case $p\ge k/4$, \texttt{\LongNestedPathName} then finds a path of length at least $f(k)/4$, and the path $Q$ returned by \texttt{\NestedDecompressionName} is of length at least $\delta(G-\{s,t\})+f(k)/32-3$.
	
	\textbf{$G_h$ is not decomposed.}
	Here, our goal is to show that taking in $G_h$ one of the $(s_h, t_h)$-paths computed on Line~\ref{line:longest_P_h} together with an arbitrary connection from $\{s, t\}$ to $\{s_h, t_h\}$ gives a long enough $(s,t)$-path in $G$.
	Let $k_h:=|E(T_h)|-\delta_h$.
	Note that by the choice of $T$, $T_h$ is the longest $(s_h,t_h)$-path in $G_h$.
	We first show the following.
	
	\begin{claim}\label{claim:not_decomp_good_approx}
		If $G_h$ is not decomposed,
		then at Line~\ref{line:outer_appendage} the length of $P_h$ is at least $\delta_h+f(k_h)/8-3$.
	\end{claim}
	\begin{claimproof}
		First note that if the length of $P_h$ from definition of \nestedbananadec is at least $|V(G_h)|-1$, then $P_h$ is a hamiltonian $(s_h,t_h)$-path in $G_h$.
		Then its length is maximum possible and is equal to $|E(T_h)|=\delta_h+k_h\ge \delta_h+f(k_h)$.
		Hence, we can assume that the length of $P_h$ given by \nestedbananadec is at least $\frac{5}{4}\delta_h-3$.
		
		If $f(k_h)\ge \frac{8}{7}\delta_h$, then \texttt{long\_st\_path\_approx}($G_h$, $s_h$, $t_h$), the blackbox $(s, t)$-path approximation algorithm, returns a path of length at least $f(\delta_h+k_h)\ge f(k_h)\ge \delta_h+f(k_h)/8$, so we are done.

		Otherwise, $f(k_h)\le \frac{8}{7}\delta_h$.
		If $f(k_h)\le 24$, then it suffices for $P_h$ to have length $\delta_h$.
		In this case \texttt{long\_eg\_st\_path}($G_h$, $s_h$, $t_h$), the exact $(s,t)$-path algorithm from \Cref{thm:relaxed_st_path}, returns a $(s_h,t_h)$-path of length at least $\delta_h$.

		It only remains to deal with the case where $f(k_h)>24$, and $\delta_h\ge \frac{7}{8} f(k_h)$.
		Since $G_h$ is not decomposed and $\delta_h>16$, by definition of \nestedbananadec, $P_h$ is of length at least $$\frac{5}{4}\delta_h-3\ge \delta_h+\frac{1}{4}\delta_h-3\ge \delta_h+\frac{1}{4}\cdot \frac{7}{8}f(k_h)-3\ge\delta_h+\frac{1}{8}f(k_h)-3.$$
	\end{claimproof}
	
	Using \eqref{eq:outside_path_bound}, we can lower-bound $k_h$ by
	\begin{equation}
	k_h=|E(T)|-|E(T)\setminus E(G_h)|-\delta_h\ge|E(T)|-3p-\sum_{i\in X}d_{j_{i+1}}-\delta_h=\delta+k-3p-\sum_{i\in X}d_{j_{i+1}}-\delta_h.
	    \label{eq:k_h_bound}
	\end{equation}
	On Line~\ref{line:outer_appendage}, $P_h$ is transformed into an $(s,t)$-path in $G$ of length at least $|E(P_h)|+\sum_{i\in X\cup Y}d_{j_{i+1}}$, by \Cref{claim:at_least_sumx}.
	By \Cref{claim:not_decomp_good_approx}, this length is at least
	\begin{align*}
		\delta_h+f(k_h)&/8-3+\sum_{i\in X\cup Y}d_{j_{i+1}}\\
		&\ge\delta+(\delta_h-\delta)+f\left(k+\delta-3p-\sum_{i\in X}d_{j_{i+1}}-\delta_h\right)/8+\sum_{i\in X\cup Y}d_{j_{i+1}}-3 \\
		&\ge \delta+\underbrace{(\delta_h+\sum_{i\in X\cup Y}d_{j_{i+1}}-\delta)}_{\ge 0\text{ by \Cref{lemma:st_path_banana_to_2_connected}}}+f\left((k-3p)-(\delta_h+\sum_{i\in X}d_{j_{i+1}}-\delta)\right)/8-3\\
		&\ge \delta+f\left(\delta_h+\sum_{i\in X\cup Y}d_{j_{i+1}}-\delta\right)+f\left((k-3p)-(\delta_h+\sum_{i\in X}d_{j_{i+1}}-\delta)\right)/8-3 \\
		&\ge \delta +f(k-3p)/8-3\\
		&\ge \delta+f(k/4)/8-3\ge \delta+f(k)/32-3.
	\end{align*}
	Here for the first inequality we use \eqref{eq:k_h_bound}, then the properties of the function $f$, and the fact that we are in the case where $k > 4p$. Observe that for each $x \in \mathbb{Z}_{+}$, $f(x) \le x$, since we are given an algorithm that finds an $(s, t)$-path of length $f(x)$ in any graph with the longest $(s, t)$-path of length $x$.
	With this, we have shown that in each case the returned $(s, t)$-path is of desired length, and the proof is complete.
\end{proof}

Finally, observe that the running time of \Cref{alg:long_nested_path} is polynomial in the size of the given \nestedbananadec. By \Cref{lemma:nested_construction}, its size is polynomial in the size of the input graph $G$. This concludes the proof of \Cref{thm:eg_approx}.

\section{Approximation for cycles}\label{sec:cycle}

This section is devoted to the proof of \Cref{thm:long_cycle_approx} that establishes a way of lifting the approximation guarantee of the longest cycle in a $2$-connected graph $G$ to the offset above $2\delta(G)$.

\diracappxtheorem*

We first recall the concept of a \cyclebananadec and its properties that were established by Fomin et al.~\cite{FominGSS22}. Further in this section, we prove the novel crucial property that in a graph admitting a \cyclebananadec, there exists a separating pair of vertices with a long path between this pair.
Finally, we combine these results together with our approximation for $(s,t)$-path from the previous section (\Cref{thm:eg_approx}) to obtain the lifting algorithm in \Cref{thm:long_cycle_approx}.

\subsection{\Cyclebananadec}

This subsection contains the definition and properties of a \cyclebananadec, including the algorithmic result that allows to construct a \cyclebananadec from a given 2-connected graph $G$.
We start with the definition of a \cyclebananadec, which can be seen as an analogue of an \bananadec for cycles.

\begin{definition}[{\Cyclebananadec and \cyclebanana}, Definition 5 in \cite{FominGSS22}]
	Let $G$ be a 2-connected graph and let $C$ be a cycle in $G$ of length at least $2\delta(G)$.
	We say that  two disjoint paths $P_1$ and $P_2$ in $G$ induce \emph{a \cyclebananadec   for   $C$}  in $G$ if
	\begin{itemize}
		\item 
		The cycle $C$ is of the form $C=P_1 {P'}P_2{P''}$, where  each of the paths ${P'}$ and ${P''}$ has at least $\delta(G)-2$ edges.
		\item 
		For  every connected component $H$ of $G-V(P_1  \cup P_2)$ holds $|V(H)|\ge 3$ and one of the following. 
		\begin{enumerate}[label=(D\arabic*)]
			\item\label{enum:cycle_tunnel_path_bic} $H$ is $2$-connected, the maximum size of a matching in  $G'$ between $V(H)$ and $V(P_1)$  is one,  and between $V(H)$ and $V(P_2)$ is also  one;
			\item\label{enum:cycle_tunnel_path_cut_left} $H$ is not 2-connected,   
			exactly one vertex of $P_1$ has neighbors in $H$, that is, 			
			$|N_{G}(V(H))\cap V(P_1)|=1$, and no inner vertex from a  leaf-block of $H$ has a neighbor in $P_2$;
			\item\label{enum:cycle_tunnel_path_cut_right} The same as  \ref{enum:cycle_tunnel_path_cut_left}, but with $P_1$ and $P_2$ interchanged. That is, 
			$H$ is not 2-connected,  			
			$|N_{G}(V(H))\cap V(P_2)|=1$, and no inner vertex from  a leaf-block of $H$ has a neighbor in $P_1$.			
		\end{enumerate}
		
		\item There is exactly one connected component $H$ in $G-V(P_1\cup P_2)$ with $V(H)=V(P')\setminus \{s',t'\}$, where $s'$ and $t'$ are the endpoints of $P'$.
		Analogously, there is exactly one connected component $H$ in $G-V(P_1\cup P_2)$ with $V(H)=V(P'')\setminus \{s'',t''\}$.
	\end{itemize}
	The set of \emph{\cyclebanana}s for a \cyclebananadec 
	is defined as follows.
	First,  for each component $H$ of type \ref{enum:cycle_tunnel_path_bic}, $H$ is a \cyclebanana of the \cyclebananadec.
	Second, for each leaf-block of each $H$ of type \ref{enum:cycle_tunnel_path_cut_left}, or of type \ref{enum:cycle_tunnel_path_cut_right}, this leaf-block  is also  a \cyclebanana of the \cyclebananadec.	
\end{definition}

First, we recall an important property of a \cyclebananadec that restricts how a cycle can pass through a \cyclebanana.

\begin{lemma}[Lemma~17 in \cite{FominGSS22}]\label{lemma:dirac_cycle_banana_consecutive}
	Let $G$ be a $2$-connected graph and $C$ be a cycle in $G$. Let paths  $P_1, P_2$ induce a \cyclebananadec  for $C$ in $G$.
	Let $M$ be a \cyclebanana of the \cyclebananadec and $P$ be a path in $G$ such that $P$ contains at least one vertex in $V(P_1)\cup V(P_2)$.
	If $P$ enters $M$, then all vertices of $M$ hit by $P$ appear consecutively on $P$.
\end{lemma}

We now restate the result of \cite{FominGSS22} on the construction of a \cyclebananadec for a given graph.
Note that here we state it in a slightly different form, which is more convenient in the setting of this paper.
The statement is given below and the main difference is highlighted in bold.

\begin{lemma}[Lemma 20 in \cite{FominGSS22}]\label{lemma:main_cycle_lemma}
	Let $G$ be an $n$-vertex $2$-connected graph and $k$ be an integer such that  $\delta(G)\ge 12$,  $0< k \le \frac{1}{24}\delta(G)$, and 
	\[2k+12\leq \delta(G)<\frac{\textbf{n}}{\textbf{2}}.
	\]
	Then there is an algorithm that, given a \textbf{non-hamiltonian} cycle $C$ of length less than $2\delta(G)+k$ in polynomial time finds either
	\begin{itemize}
		\item Longer cycle in $G$, or
		\item Vertex cover of $G$ of size at most $\delta(G)+2k$, or
		\item Two paths $P_1, P_2$ that induce a \cyclebananadec for $C$ in $G$.
	\end{itemize}
\end{lemma}

To clarify this form, note that in the original statement of Lemma~20 in \cite{FominGSS22}, the upper bound for $\delta(G-B)$ is $\frac{n}{2}-\frac{|B|+k}{2}$, where $B$ is a given set of small-degree vertices.
In the proof of Lemma~20 in \cite{FominGSS22}, one can easily note that the only reason for this bound is the existence of at least one vertex in $V(G)\setminus V(C) \setminus B$.
Since in our work $B$ is always empty, this is equivalent to $V(G)\neq V(C)$, i.e.\ non-hamiltonicity of $C$.
Hence, the replacement of this bound with the condition on non-hamiltonicity of $C$ is legitimate.

We finish this subsection with another important property of \cyclebananadec stating that there is a long cycle that enters at least one \cyclebanana.
Unfortunately, this property, Lemma~19 in \cite{FominGSS22}, is stated in a way requiring the offset above $2\delta(G)$ for this long cycle to be much smaller than $\delta(G)$.
Here we provide this property in the form that does not require this and is much more convenient in our setting.
Since it differs significantly from the original statement, we provide a proof of this result that is based on the proof of Lemma~19 from \cite{FominGSS22}.

\begin{lemma}[Modified Lemma~19 from \cite{FominGSS22}]\label{lemma:dirac_cycle_edge_of_banana}
	Let $G$ be a graph and $P_1, P_2$ induce a \cyclebananadec  for a cycle $C$ of length at most $2\delta(G-B)+\kappa$ in $G$ such that $2\kappa \le \delta(G)$.
	If there exists a cycle of length at least $2\delta(G)+k$ in $G$ that contains at least one vertex in $V(P_1)\cup V(P_2)$, then there exists a cycle of length at least $2\delta(G)+k/2-1$ in $G$ that  	
	enters a \cyclebanana.
\end{lemma}
\begin{proof}
	Suppose that there exists a cycle $C'$ of length at least $2\delta(G)+k$ in $G$ that contains at least one vertex in $V(P_1)\cup V(P_2)$.
	If $C'$ already contains an edge of a \cyclebanana, we are done.
	We now assume that $C'$ does not contain any edge of any \cyclebanana.
	We show how to use $C'$ to construct a cycle of length at least $2\delta(G)+k/2 - 1$ in $G$ that contains an edge of a \cyclebanana of the given \cyclebananadec.
	
	Let $W$ be the set of all vertices of $G$ that are vertices of non-leaf-blocks of \ref{enum:cycle_tunnel_path_cut_left}-type or \ref{enum:cycle_tunnel_path_cut_right}-type components in the \cyclebananadec. 
	We start with the following claim.
	
	\begin{claim}
		$|W\cap V(C')|>0$.
	\end{claim}
	\begin{claimproof}
		This is a counting argument.
		Note that $C'$ cannot contain an edge with both endpoints inside a \cyclebanana of $G$.
		Since \cyclebanana{s} of $G$ are \ref{enum:cycle_tunnel_path_bic}-type components of the \cyclebananadec and leaf-blocks of \ref{enum:cycle_tunnel_path_cut_left}-type or \ref{enum:cycle_tunnel_path_cut_right}-type connected components, each edge of $C'$ has an endpoint either in $V(P_1)\cup V(P_2)$, or inside a non-leaf-block of a \ref{enum:cycle_tunnel_path_cut_left}-type or a \ref{enum:cycle_tunnel_path_cut_right}-type connected component.
		The union of the vertex sets of the non-leaf-blocks forms the set $W$.
		Hence, $(W\cap V(C'))\cup V(P_1)\cup V(P_2)$ is a vertex cover of $C'$.
		
		Note that a vertex cover of any cycle consists of at least half of its vertices.
		Then $$2|(W\cap V(C'))\cup V(P_1)\cup V(P_2)|\ge |V(C')|.$$
		By definition of a \cyclebananadec, $|V(P_1)\cup V(P_2)|\le \kappa-2$.
		Immediately we get that $$2|W\cap V(C')|\ge 2\delta(G)+k-2|V(P_1)\cup V(P_2)|\ge 2\delta(G)+k-2(\kappa-2)>0.$$
	\end{claimproof}
	
	We now take a vertex $w_1 \in W\cap V(C')$.
	The following claim allows constructing a long chord of $C'$ starting in $w_1$.
	
	\begin{claim}\label{claim:hard_lemma_no_inner_vertices}
		Let $H$ be a \ref{enum:cycle_tunnel_path_cut_left}-type or a \ref{enum:cycle_tunnel_path_cut_right}-type component in the \cyclebananadec.
		$C'$ does not contain any inner vertex of the leaf-blocks of $H$.
	\end{claim}
	\begin{claimproof}
		Suppose that $C'$ contains some vertex $u \in V(H')$ that is an inner vertex of some leaf-block $L$ of $H$.
		As $L$ is a \cyclebanana of $G$, $C'$ cannot contain any edge of $L$, so $C'$ should enter $L$ from $V(P_1) \cup V(P_2)$ through $u$ and leave it immediately.
		By definition of \cyclebananadec{s}, the only option to enter or leave $L$ is to go through the only vertex in $V(P_1)$ (if $H$ is of type \ref{enum:cycle_tunnel_path_cut_left}) or in $V(P_2)$ (if $H$ is of type \ref{enum:cycle_tunnel_path_cut_right}).
		As $C'$ cannot contain any vertex twice, this is not possible.
	\end{claimproof}

	Now construct the chord of $C'$ starting in $w_1$.
	Since $w_1$ is a vertex of a  separable component $H$, there is a cut vertex $c_1$ of a leaf-block $L_1$ of $H$ reachable from $w_1$ inside $H$.
	The leaf-block $L_1$ contains also at least one vertex $v_1\neq c_1$ that is adjacent to a vertex in $V(P_1)$ (if $H$ is of type \ref{enum:cycle_tunnel_path_cut_left}) or to $V(P_2)$ (if $H$ is of type \ref{enum:cycle_tunnel_path_cut_right}) outside $H$.
	We know that $\delta(L_1-c_1)\ge \delta(G-c_1)-1\ge \delta(G)-2$, since the only outside neighbour of vertices in $L_1-c_1$ is a single vertex in $V(P_1)$ or $V(P_2)$.
	By Corollary~\ref{thm:relaxed_st_path}, there exists an $(c_1,v_1)$-path inside $L_1$ of length at least $\delta(G)-2$.
	Combine this with a $(w_1,c_1)$-path inside $H$ and obtain a $(w_1,v_1)$-path inside $H$.
	
	Note that the constructed $(w_1,v_1)$-path can contain vertices from $V(C')$ apart from $w_1$.
	Let $w'_1 \in V(C')$ be the vertex from $V(C')$ on the $(w_1,v_1)$-path farthest from $w_1$.
	Note that the $(w'_1,v_1)$-subpath still contains the $(c_1,v_1)$-path as a subpath by Claim~\ref{claim:hard_lemma_no_inner_vertices}.
	Hence, we obtain a $(w'_1,v_1)$-path of length at least $\delta(G)-2$ inside $H$ that does not contain any vertex in $V(C')\setminus \{w'_1\}$.
	To obtain a long chord of $C'$, it is left to reach the vertex in $V(P_1)\cup V(P_2)$ from $v_1$ outside $H$, and then follow the cycle $C$ until a vertex $v'_1$ of $C'$ is reached. This is always possible since $V(C)\cap V(C')\supseteq (V(P_1)\cup V(P_2))\cap V(C')\neq \emptyset$.
	We obtain a chord of length at least $\delta(G)-1$ connecting $w'_1$ and $v'_1$.
	
	The $(w'_1,v'_1)$-chord of $C'$ splits $C'$ into two $(w'_1,v'_1)$-arcs, and one of the arcs has length at least $\delta(G)+k/2$.
	Combine this arc with the chord and obtain a cycle of length at least $2\delta(G)+k/2-1$ in $G$.
	This cycle contains an edge of a leaf-block of $H$, i.e.\ of a \Cyclebanana.
	The proof is complete.
\end{proof}

\subsection{Existence of a separating pair}

This subsection encapsulates the new combinatorial result behind \cyclebananadec that is crucial to our proof of \Cref{thm:long_cycle_approx}.
It helps us avoid using the \cyclebananadec explicitly in our algorithm, so that we can instead reduce to the algorithm for approximating $(s,t)$-paths. The formal statement is recalled next.
\lemsinglepath*
\begin{proof}
	Consider the longest cycle $C'$ in $G$.
	We assume that this cycle is of length at least $2\delta(G)+k$ and consider four cases.

	\textbf{Case 1.} \emph{$C'$ is completely contained in some connected component $H$ of $G-V(P_1\cup P_2)$, and $H$ is $2$-connected.}
	Then $H$ is a \Cyclebanana of type \ref{enum:cycle_tunnel_path_bic}.
	Since the matching size between $V(H)$ and $V(P_i)$ for each $i\in\{1,2\}$ is exactly one, by K\H{o}nig's theorem all edges between $V(H)$ and $V(P_i)$ are covered by a single vertex.
	Denote this vertex by $u$ for $i=1$ and by $v$ for $i=2$.
	Since $G$ is $2$-connected, $u$ and $v$ are distinct.
	As $\{u,v\}$ separates $H$ from the rest of the graph and $V(C)\not\subset V(H\cup P_1\cup P_2)$, we have that $G-\{u,v\}$ is not connected.
	It is left to show that there exists a long $(u,v)$-path in $G$.
	Toward this, denote $u'=u$ if $u\in V(H)$, and $u'\in N_G(u)\cap V(H)$ if $u \in V(P_1)$.
	Choose $v' \in V(H)$ similarly, i.e.\ $v'$ either equals $v$ or is a neighbour of $v$.
	Since $G$ is $2$-connected, there is always a way to choose distinct $u'$ and $v'$.
	By \Cref{lemma:cycle_to_path}, there is a $(u',v')$-path of length at least $\delta(G)+k/2$ in $H$, hence there is also a $(u,v)$-path of length at least $\delta(G)+k/2$ in $G$.
	
	\textbf{Case 2.} \emph{$C'$ is completely contained in a leaf-block of a connected component $H$ of $G-V(P_1\cup P_2)$.}
	That is, $C'$ is contained in a \cyclebanana of type \ref{enum:cycle_tunnel_path_cut_left} or \ref{enum:cycle_tunnel_path_cut_right}.
	The choice of $u$ and $v$ is similar to Case 1.
	That is, if the \cyclebanana is of type \ref{enum:cycle_tunnel_path_cut_left}, choose $u$ such that $u \in N_G(V(H))\cap V(P_1)$ and choose $v$ equal to the cut vertex of the \cyclebanana.
	$G-\{u,v\}$ is not connected as $\{u,v\}$ separates the \cyclebanana from the rest of $H$.
	There is an $(u,v)$-path of length at least $\delta(G)+k/2+1$ since there exists a $(z,v)$-path of length at least $\delta(G)+k/2$ by \Cref{lemma:cycle_to_path}, where $z\in N_G(u)\cap V(H-v)$.
	The choice of $u$ and $v$ for type \ref{enum:cycle_tunnel_path_cut_right} is symmetrical.
	
	\textbf{Case 3.} \emph{$C'$ is completely contained in a non-leaf-block of a connected component $H$ of $G-V(P_1\cup P_2)$.}
	In this case, $C'$ is not contained in a \cyclebanana.
	Denote the non-leaf-block of $H$ that contains $C'$ by $K$.
	Without loss of generality, we assume that $H$ corresponds to \ref{enum:cycle_tunnel_path_cut_left}, i.e.\ $|N_G(V(H))\cap V(P_1)|=1$.
	By Menger's theorem, there are either two vertices separating $V(K)$ from $V(P_1\cup P_2)$ in $G$ or three disjoint paths going from $V(K)$ to $V(P_1\cup P_2)$.
	If the former is the case, denote these two vertices by $u$ and $v$.
	Obviously, $G-\{u,v\}$ is not connected.
	There are two disjoint paths going from $V(K)$ to $V(P_1\cup P_2)$, and one of these paths contains $u$ and the other contains $v$.
	Connect the endpoints of these paths in $V(K)$ using a path of length at least $\delta(G)+k/2$ inside $V(K)$ given by \Cref{lemma:cycle_to_path}.
	Clearly, the obtained long path contains an $(u,v)$-path of length at least $\delta(G)+k/2$ as a subpath.
	
	If the latter is the case, then two of the three paths necessarily end in $V(P_2)$.
	Let these two paths start respectively from $u_1$ and $u_2$ in $V(K)$ and end in $v_1$ and $v_2$ in $V(P_2)$.
	Note that these paths use only edges of $H$ and edges between $V(H)$ and $V(P_2)$ in $G$.
	Moreover, none of the paths has an internal vertex in $V(K)$ or $V(P_2)$.
	By \Cref{lemma:cycle_to_path}, there is an $(u_1,u_2)$-path of length at least $\delta(G)+k/2$ in $K$.
	Now construct a cycle in $G$ by combining the $(u_1,u_2)$-path with the $(u_1,v_1)$-path, $(u_2,v_2)$-path, and the subpath of $C$ that goes between $v_1$ and $v_2$ outside of $P_2$.
	Since that subpath contains $P'$ and $P''$ from the definition of \Cyclebananadec as subpaths, the obtained cycle is of length at least $(\delta(G)+k/2)+1+1+2\cdot(\delta(G)-2)\ge 3\delta(G)+k/2-2\ge 2\delta(G)+k/2$.
	This cycle contains an edge of $P'$, so it enters a \Cyclebanana, so we can replace $C'$ with this cycle and apply the following case.
	
	\textbf{Case 4.} \emph{$C'$ has a common vertex with $V(P_1\cup P_2)$.
	By \Cref{lemma:dirac_cycle_edge_of_banana}, we can assume that $C'$ enters a \Cyclebanana $K$ but its length is at least $2\delta(G)+k/2-1$.}
	Following Case 1 and Case 2, we know that there are $u, v \in V(K\cup P_1 \cup P_2)$ such that in $G-\{u,v\}$ vertices in $V(K)$ are separated from the rest of the graph.
	By \Cref{lemma:dirac_cycle_banana_consecutive}, vertices in $V(K)$ appear consecutively on $C'$, so vertices and edges of $C'$ induce a path inside $K$.
	Since $C'$ is not contained in $V(K)$, at least one of $u$ and $v$ is present in $C'$, so we have two cases depending on $|V(C')\cap\{u,v\}|$.
	If $u,v \in V(C')$, then the longest arc of $C'$ going between $u$ and $v$ simply yields a path of length at least $(2\delta(G)+k/2-1)/2= \delta(G)+(k-2)/4$.
	If exactly one of $u$ and $v$ is present on $C'$, without loss of generality we assume $u \in V(C')$.
	Then $V(C')\subseteq V(K)\cup \{u\}$, as $C'$ does not pass through $v$ ---  the only other entry to $K$.
	Similarly to Case 1 and Case 2, we have a vertex $v' \in V(K)$, which is either equal to $v$ or is a neighbour of $v$.
	Take a shortest path from $v'$ to $V(C')$ inside $K$.
	Denote its endpoint by $w$.
	Prolong the path starting in $v'$ with the longest arc of $C'$ that goes between $w$ and $u$.
	This yields a $(v',u)$-path, hence a $(v,u)$-path, of length at least $\delta(G)+(k-2)/4$ in $G$.
\end{proof}

\subsection{Proof of \Cref{thm:long_cycle_approx}}

In this subsection, we combine \Cref{thm:eg_approx} and the results presented earlier in this section  into the proof of \Cref{thm:long_cycle_approx}.

\begin{proof}[Proof of \Cref{thm:long_cycle_approx}]
	Assume that we are given a blackbox algorithm that finds a cycle of length $f(L)$ in a graph with the longest cycle length $L$.
	We now describe the desired approximation algorithm that finds a cycle of length at least $2\delta(G)+h(k)$ based on the blackbox algorithm, where $$h(k)=\frac{1}{128}f(k)-8.$$
	
	The input to our algorithm is a graph $G$,
	let $L$ be the length of the longest cycle in $G$ and $k=L-2\delta$.
	For convenience, denote $\delta:=\delta(G)$.
	The goal of our algorithm is to find a cycle of length at least $2\delta+h(k)$ in $G$.
	Note that the algorithm does not estimate $h(k)$ in any way, it merely outputs the longest cycle that was found during its run.
	We focuse on showing that this cycle always has length at least $2\delta+h(k)$.
	
	The pseudocode of our algorithm is presented in \Cref{alg:long_cycle}.
	The first few lines of the algrotihm are dedicated to eliminating various corner cases where either the blackbox approximation suffices directly, or a long Dirac cycle. This will help us avoid dealing with extreme parameter values later in the analysis.

	If $2 \delta \ge n$, the algorithm will find and output a Hamiltonian cycle in $G$ following Dirac's theorem on Line~\ref{line:return_hamiltonian}. For the rest of the analysis, we assume $2\delta < n$.
	On Line~\ref{line:cycle_blackbox} our algorithm applies the blackbox $f(L)$-approximation algorithm to $G$.
	If $f(L) \ge \frac{49}{24}\delta$, then the resulting cycle is of length at least $2\delta+(f(L)-2\delta)\ge 2\delta+\frac{1}{49}f(L)$, which is at least $2 \delta + h(k)$. As the algorithm never makes the current cycle shorter, in this case the output will be automatically valid.
	We now also assume that $f(L)<\frac{49}{24}\delta$.
	
    \IncMargin{1em}
	\begin{algorithm}[!h]
		\SetKwFunction{LongestPath}{longest\_path}
		\SetKwFunction{LongCycleApprox}{longest\_cycle\_approx}
		\SetKwFunction{LongSTPath}{long\_st\_path}
		\SetKwFunction{LongSTPathApprox}{longest\_st\_path\_above\_degree\_approx}
		\SetKwFunction{LongErdosSTPath}{long\_eg\_st\_path}
		\SetKwFunction{LongSTCycle}{long\_st\_cycle}
		\SetKwFunction{LongDiracCycle}{long\_dirac\_cycle}
		\SetKwFunction{LongDiracCycleApprox}{longest\_cycle\_above\_degree\_approx}
		\SetKwFunction{HamPath}{hamiltonian\_path}
		\let\oldnl\nl
		\newcommand{\nonl}{\renewcommand{\nl}{\let\nl\oldnl}}
		\Indm \nonl \LongDiracCycleApprox{$G$}
		
		\Indp
		
		\KwIn{$2$-connected graph $G$ of minimum degree $\delta$}
		\KwOut{a cycle $C$ of length at least $2\delta+h(k)$, where $k=\cyclelength-2\delta$ and \mathcyclelength is the length of the longest cycle in $G$}
		
		$C \longleftarrow \LongCycleApprox(G)$\label{line:cycle_blackbox}\;
		
		\eIf{$\LongDiracCycle(G,\emptyset,1)$is \textsc{Yes}\label{line:apply_long_dirac}}{
		    $C\longleftarrow \text{the longest of $C$ and the computed cycle of length at least } 2\delta+1 \text{ in } G$\;
		}{
		    \Return the cycle of length $2\delta$ in $G$\label{line:return_hamiltonian}\;
		}
		
		\If{$\delta\le 24$\label{line:check_delta}}{
			\Return $C$\;
		}
		
		\label{line:cycle_enlarging}\While{$|V(C)|-2\delta< \lfloor\frac{1}{24}\delta\rfloor$ and \Cref{lemma:main_cycle_lemma} applied to $G$ and $C$ gives a longer cycle}{
			$C \longleftarrow \text{ a longer cycle in } G$\;
		}
		
		\If{$|V(C)|\ge \lfloor 2\frac{1}{24}\delta\rfloor$ or \Cref{lemma:main_cycle_lemma} gives the vertex cover of $G$}{
			\Return $C$;
		}
		
		\ForEach{$u,v \in V(G)$ such that $G-\{u,v\}$ is not connected}{
			$Q,R\longleftarrow $ empty paths\;
			\ForEach{connected component $H$ in $G-\{u,v\}$}{$S \longleftarrow \LongSTPathApprox(G[V(H)\cup \{u,v\}]+uv,u,v)$\;
				$Q,R \longleftarrow $ two longest paths among ${Q,R,S}$\;
			}
			$C \longleftarrow $ the longest of $C, Q\cup R$\;
		}
		\Return $C$\;
		\caption{The algorithm finding a cycle of length at least $2\delta(G)+h(k)$ in a $2$-connected graph $G$.}
		\label{alg:long_cycle}
	\end{algorithm}
    \DecMargin{1em}

	On Line~\ref{line:apply_long_dirac} our algorithm applies the FPT algorithm for \probDC to find a cycle of length at least $2\delta+1$ in $G$ in polynomial time.
	If such cycle is found, then our algorithm keeps the longest of this cycle and previously computed approximation.
	If $h(k)=\frac{1}{128}f(k)-8\le 1$, then this cycle is a required approximation.
	On the other hand,  if a cycle of length at least $2\delta+1$ does not exist in $G$, then $k=L-2\delta=0$, so a cycle of length $2\delta$ is a valid approximation. Hence, in this case our algorithm just outputs a cycle of length at least $2\delta$ guaranteed by Dirac's theorem in this case and stops.
		
	We now can assume that $f(k)\ge 9 \cdot 128$.
	Since $f(L)<\frac{49}{24}\delta$, it follows that $\delta>24$.
	Thus, on Line~\ref{line:check_delta}, if $\delta\le 24$, our algorithm just stops as the required approximation cycle was already encountered by the algorithm.

	Now we reach the main case of the algorithm, where we use the structural results on \cyclebananadec to find a long cycle.
    Before Line~\ref{line:cycle_enlarging}, the current cycle $C$ has length at least $2\delta$.
	If the length of $C$ is less than $\lfloor 2\frac{1}{24}\delta\rfloor$, then the algorithm of \Cref{lemma:main_cycle_lemma} is applied to the graph $G$ and the cycle $C$ with the parameter $k'=|V(C)|-2\delta+1$.
	If the outcome is a cycle longer than $C$, then it replaces $C$ with this cycle.
	If it still holds that $|V(C)|<\lfloor 2\frac{1}{24}\delta\rfloor$, then our algorithm applies \Cref{lemma:main_cycle_lemma} to $G$ and $C$ again.
	This process repeats until one of the three possible structures are found in $G$:
	\begin{itemize}
		\item Cycle $C$ of length at least $\lfloor2\frac{1}{24}\delta\rfloor$;
		\item Vertex cover of size at most $\delta+2(|V(C)|-2\delta+1)$;
		\item Two paths $P_1, P_2$ that induce a \cyclebananadec for $C$.
	\end{itemize}
	
	The first outcome is the desired $h(k)$-approximation of the offset since $|V(C)|\ge 2\delta+\frac{1}{49}f(L)-1\ge 2\delta+h(k)$ in this case.
	In the second outcome, since a vertex cover upper-bounds the length of any cycle, $L \le 2\cdot (2\cdot |V(C)|-3\delta+2)$, hence $|V(C)|-2\delta \ge \frac{L-4-2\delta}{4}$, so $|V(C)| \ge 2\delta+\frac{k}{4}-1$.
	Automatically, $C$ is a valid approximation in this case as well. Thus if any of the two situations occur, our algorithm simply returns the current cycle $C$.
	
	We move on to the most involved case where the two paths $P_1, P_2$ inducing a \cyclebananadec are found.
	Our goal now is to use \Cref{lemma:single_path_lemma} to find a separating pair of vertices that has a long path between them, and then use 
	the already established algorithm from \Cref{thm:eg_approx} to approximate the length of this path.
	
	Hence, before moving further, we need to obtain an approximation algorithm for finding a long $(s,t)$-path.
	For this, we apply \Cref{lemma:stpath_approx} to the blackbox $f(L)$-approximation algorithm for the longest cycle and obtain an algorithm that finds an $(s,t)$-path of length at least $\frac{1}{2}f(2p)$ in a graph with the longest $(s,t)$-path lenght $p$.
	Finally, we apply \Cref{thm:eg_approx} to the latter algorithm and obtain an algorithm finding an $(s,t)$-path of length $\delta+(\frac{1}{64}f(2k')-3)$ where $k'=p-\delta$ for the longest $(s,t)$-path length $p$.
	
	Since $2(|V(C)|-2\delta)< \delta$, we can apply \Cref{lemma:single_path_lemma} to $G$ and $C$.
	We obtain that there exists a pair of vertices $u,v\in V(G)$ such that $G-\{u,v\}$ is not connected and there is a path of length at least $\delta+(L-2\delta-2)/4$ between $u$ and $v$ in $G$.	
	Towards encountering this pair of vertices, our algorithm iterates over all possible $u,v \in V(G)$ such that $G-\{u,v\}$ is not connected.
	Assume that the pair $\{u, v\}$ is fixed.
	Then for each connected component $H$ of $G-\{u,v\}$, our algorithm applies the $(s, t)$-path approximation algorithm to $G[V(H)\cup \{u,v\}]+uv$ to find a long $(u,v)$-path.
	Note that this is a legitimate application of the algorithm since $G[V(H)\cup \{u,v\}]+uv$ is $2$-connected.
	
	Now, each application of the algorithm yields some $(u,v)$-path.
	There are at least two connected components in $G-\{u,v\}$, so at least two $(u, v)$-paths are produced, and our algorithm simply combines the longest two paths among them in a cycle.
	The length of each path is at least $\delta-2$.
	Moreover, if $\{u,v\}$ is the pair given by \Cref{lemma:single_path_lemma}, for at least one component the longest $(u,v)$-path has length at least $\delta+(k-2)/4$.
	Then for this connected component $H$ the approximation algorithm finds an $(u,v)$-path $Q$ of length at least
	$\delta(H-\{u,v\})+\frac{1}{64}f(2k')-3$, where $$k'=\delta(G)+(k-2)/4-\delta(H-\{u,v\}).$$
	
	Denote $x=(\delta(G)-\delta(H-\{u,v\})$.
	Note that $x\le 2$ and can be negative.
	Then the length of $Q$ is at least
	$$\delta(G)-x+\frac{1}{64}f(2(x+(k-2)/4))-3.$$
	If $x\ge 0$, then the length of $Q$ is at least
	$$\delta(G)-2+\frac{1}{64}f((k-2)/2)-3.$$
	If $x<0$, then, as $-x\ge f(|x|)$, the length of $Q$ is at least
	\begin{align*}
		\delta(G)&+f(|x|)+\frac{1}{64}f(-2|x|+(k-2)/2)-3 \\
		&\ge \delta(G)+\frac{1}{64}f(64|x|)+\frac{1}{64}f(-2|x|+(k-2)/2)-3 \\
		&\ge \delta(G)+\frac{1}{64}f(62|x|+(k-2)/2)-3 \\
		&\ge \delta(G)+\frac{1}{64}f((k-2)/2)-3.
	\end{align*}	
	In any case, the length of $Q$ is at least $$\delta+\frac{1}{64}f(k/2-1)-5\ge \delta+\frac{1}{128}f(k-2)-5\ge \delta+\frac{1}{128}f(k)-\frac{1}{128}f(2)-5\ge \delta+\frac{1}{128}f(k)-6.$$
	Hence, if our algorithm combines the longest two paths given by the pair $\{u,v\}$, it obtains a cycle of length at least $2\delta+\frac{1}{128}f(k)-8=2\delta+h(k)$.
	Since the algorithm outputs the longest cycle among those constructed, the proof is complete.
\end{proof}

Finally, we show the corollary of \Cref{thm:long_cycle_approx} that allows to  approximate the longest path in a graph that is not necessarily 2-connected.
\begin{corollary}\label{thm:long_path_approx}
	Let $f: \mathbb{R}_+\to \mathbb{R}$ be a non-decreasing subadditive function.
	If there exists a polynomial-time algorithm finding a cycle of length $f(\cyclelength)$ in a $2$-connected graph with the longest cycle length \mathcyclelength,
	there is a polynomial time algorithm that outputs a path of length at least $2\delta(G)+f(\pathlength-2\delta(G))/128-8$ in a graph $G$ with $\delta(G)<\frac{1}{2}|V(G)|$ and the longest path length \mathpathlength.	
\end{corollary}

\begin{proof}[Proof of \Cref{thm:long_path_approx}]
	Take a graph $G$.
	We assume that $G$ is connected, otherwise we apply the same algorithm to each of its connected components.
	
	Add a universal vertex to $G$ and obtain a $2$-connected graph $G'$.
	Note that $\delta(G')=\delta(G)+1$ and there is a cycle of length at least $c$ in $G'$ if and only if there is a path of length at least $c-2$ in $G$.
	Hence, $c-2\delta(G')=c-2\delta(G)-2=p-2\delta(G)$ for the longest cycle length $c$ in $G'$ and the longest path length $p$ in $G$.
	Moreover, a cycle of length $c$ in $G'$ can be transformed into a path of length at least $c-2$ in $G$.
	
	Apply \Cref{thm:long_cycle_approx} to $G'$.
	The obtained cycle is of length at least $2\delta(G')+f(k)/128-8$ where $k=c-2\delta(G')$ for the longest cycle length $c$ in $G'$.
	Finally, transform this cycle into a path of length at least $2\delta(G)+f(k)/128-8$ in $G$.
\end{proof}

\section{Conclusion}\label{sec:concl}

In this article, we have shown a general theorem that allows us to leverage all the algorithmic machinery for approximating the length of the longest cycle to approximate the ``offset'' of the longest cycle provided by the classical Dirac's theorem. As far as one can compute a cycle of length $f(\cyclelength)$ in a $2$-connected graph $G$ with the longest cycle length $L$, we can also construct a cycle of length $2\delta(G)+\Omega(f(\cyclelength-2\delta(G)))$.  
In particular, we can use the 
state-of-the-art approximation algorithm for Longest Cycle due to Gabow and Nie \cite{GabowN08}.
They achieve an algorithm finding a cycle of length $f(\cyclelength)=c^{ \sqrt{\log \cyclelength}}$ for some constant $c>1$ in a 
graph with the longest cycle length \mathcyclelength.
Note that $f$ is non-decreasing and subadditive (as $f$ is concave on $[1,+\infty]$, and any concave function is subadditive; we also can formally set $f(x)=\min\{x,c^{ \sqrt{\log x}}\}$ for $x\geq 1$ and $f(x)=x$ for $x<1$ to fit the statement of \Cref{thm:long_cycle_approx}).
By substituting this to \Cref{thm:long_cycle_approx}, we achieve a polynomial-time algorithm that outputs a cycle of length 
 $2\delta(G)+2^{\Omega(\sqrt{\log(L-2\delta(G))})}$ in a $2$-connected graph $G$ with the longest cycle length $L>2\delta(G)$.

In the field of parameterized algorithms, there are many results on computing longest cycles or paths above some guarantees. It is a natural question, whether approximation results similar to ours hold for other types of ``offsets''. To give a few concrete questions, recall that the \emph{degeneracy} $\dg(G)$ of a graph $G$ is the maximum $d$ such that $G$ has an induced subgraph of minimum degree $d$. By Erd{\H{o}}s and Gallai \cite{ErdosG59}, a graph of degeneracy $d\geq 2$ contains a cycle of length at least $d+1$. 
It was shown by Fomin et al. in~\cite{fomin_et_al:LIPIcs:2019:11168} that a cycle of length at least $\cyclelength=\dg(G)+k$ in a 2-connected graph can be found in $2^{\Oh(k)}\cdot n^{\Oh(1)}$ time. This immediately yields a polynomial-time algorithm for computing a cycle of length at least $\dg(G)+\Omega(\log(\cyclelength-\dg(G)))$.  Is there a better approximation of the longest cycle  above the degeneracy? 

Another concrete question.  Bez{\'{a}}kov{\'{a} et al.~\cite{BezakovaCDF17} gave an FPT algorithm that for $s,t\in V(G)$ finds a detour in an undirected graph $G$. In othere words, they gave an algorithm that finds an  $(s,t)$-path  of length at least $L=\dist_G(s,t)+k$ in $2^{\Oh(k)}\cdot n^{\Oh(1)}$ time. Here $\dist_G(s,t)$ is the distance  between $s$ and $t$.  Therefore, in undirected graph we can  find   an $(s,t)$-path of length $\dist_G(s,t)+\Omega(\log(\pathlength-\dist_G(s,t))$ in polynomial time. The existence of any better bound  is open. 
 For directed graphs, the question of whether finding a long detour  is FPT is widely open~\cite{BezakovaCDF17}. Nothing is known about the (in)approximability of long detours in directed graphs.


\begin{thebibliography}{10}

\bibitem{AlonGKSY10}
Noga Alon, Gregory Gutin, Eun~Jung Kim, Stefan Szeider, and Anders Yeo.
\newblock Solving {MAX}-$r$-{SAT} above a tight lower bound.
\newblock In {\em Proceedings of the 21st Annual ACM-SIAM Symposium on Discrete
  Algorithms (SODA)}, pages 511--517. SIAM, 2010.

\bibitem{AlonYZ95}
Noga Alon, Raphael Yuster, and Uri Zwick.
\newblock Color-coding.
\newblock {\em J. {ACM}}, 42(4):844--856, 1995.
\newblock URL: \url{http://dx.doi.org/10.1145/210332.210337}, \href
  {https://doi.org/10.1145/210332.210337} {\path{doi:10.1145/210332.210337}}.

\bibitem{BazganST99}
Cristina Bazgan, Miklos Santha, and Zsolt Tuza.
\newblock On the approximation of finding a(nother) hamiltonian cycle in cubic
  hamiltonian graphs.
\newblock {\em J. Algorithms}, 31(1):249--268, 1999.
\newblock \href {https://doi.org/10.1006/jagm.1998.0998}
  {\path{doi:10.1006/jagm.1998.0998}}.

\bibitem{BezakovaCDF17}
Ivona Bez{\'{a}}kov{\'{a}}, Radu Curticapean, Holger Dell, and Fedor~V. Fomin.
\newblock Finding detours is fixed-parameter tractable.
\newblock {\em {SIAM} J. Discrete Math.}, 33(4):2326--2345, 2019.
\newblock \href {https://doi.org/10.1137/17M1148566}
  {\path{doi:10.1137/17M1148566}}.

\bibitem{DBLP:journals/siamcomp/Bjorklund14}
Andreas Bj{\"{o}}rklund.
\newblock Determinant sums for undirected hamiltonicity.
\newblock {\em {SIAM} J. Comput.}, 43(1):280--299, 2014.
\newblock \href {https://doi.org/10.1137/110839229}
  {\path{doi:10.1137/110839229}}.

\bibitem{BjorkHus03}
Andreas Bj\"{o}rklund and Thore Husfeldt.
\newblock Finding a path of superlogarithmic length.
\newblock {\em SIAM J. Comput.}, 32(6):1395--1402, 2003.
\newblock \href {https://doi.org/10.1137/S0097539702416761}
  {\path{doi:10.1137/S0097539702416761}}.

\bibitem{BjHuKK10}
Andreas Bj{\"o}rklund, Thore Husfeldt, Petteri Kaski, and Mikko Koivisto.
\newblock Narrow sieves for parameterized paths and packings.
\newblock {\em CoRR}, abs/1007.1161, 2010.

\bibitem{BjorklundHK04}
Andreas Bj{\"{o}}rklund, Thore Husfeldt, and Sanjeev Khanna.
\newblock Approximating longest directed paths and cycles.
\newblock In {\em Proceedings of the 31st International Colloquium on Automata,
  Languages and Programming (ICALP)}, volume 3142 of {\em Lecture Notes in
  Comput. Sci.}, pages 222--233. Springer, 2004.
\newblock \href {https://doi.org/10.1007/978-3-540-27836-8\_21}
  {\path{doi:10.1007/978-3-540-27836-8\_21}}.

\bibitem{Bodlaender93a}
Hans~L. Bodlaender.
\newblock On linear time minor tests with depth-first search.
\newblock {\em J. Algorithms}, 14(1):1--23, 1993.

\bibitem{BollobasScot02}
B.~Bollob\'{a}s and A.~D. Scott.
\newblock Better bounds for {M}ax {C}ut.
\newblock In {\em Contemporary combinatorics}, volume~10 of {\em Bolyai Soc.
  Math. Stud.}, pages 185--246. J\'{a}nos Bolyai Math. Soc., Budapest, 2002.

\bibitem{MR1373679}
B\'{e}la Bollob\'{a}s.
\newblock Extremal graph theory.
\newblock In {\em Handbook of combinatorics, {V}ol. 1, 2}, pages 1231--1292.
  Elsevier Sci. B. V., Amsterdam, 1995.

\bibitem{MR1373656}
J.~A. Bondy.
\newblock Basic graph theory: paths and circuits.
\newblock In {\em Handbook of combinatorics, {V}ol. 1, 2}, pages 3--110.
  Elsevier Sci. B. V., Amsterdam, 1995.

\bibitem{chen2006approximating}
Guantao Chen, Zhicheng Gao, Xingxing Yu, and Wenan Zang.
\newblock Approximating longest cycles in graphs with bounded degrees.
\newblock {\em SIAM Journal on Computing}, 36(3):635--656, 2006.

\bibitem{CrowstonJMPRS13}
Robert Crowston, Mark Jones, Gabriele Muciaccia, Geevarghese Philip, Ashutosh
  Rai, and Saket Saurabh.
\newblock Polynomial kernels for lambda-extendible properties parameterized
  above the {P}oljak-{T}urzik bound.
\newblock In {\em IARCS Annual Conference on Foundations of Software Technology
  and Theoretical Computer Science (FSTTCS)}, volume~24 of {\em Leibniz
  International Proceedings in Informatics (LIPIcs)}, pages 43--54, Dagstuhl,
  Germany, 2013. Schloss Dagstuhl--Leibniz-Zentrum fuer Informatik.

\bibitem{cygan2015parameterized}
Marek Cygan, Fedor~V. Fomin, {\L}ukasz Kowalik, Daniel Lokshtanov, D{\'a}niel
  Marx, Marcin Pilipczuk, Micha{\l} Pilipczuk, and Saket Saurabh.
\newblock {\em Parameterized Algorithms}.
\newblock Springer, 2015.

\bibitem{Diestel}
Reinhard Diestel.
\newblock {\em Graph theory}, volume 173 of {\em Graduate Texts in
  Mathematics}.
\newblock Springer-Verlag, Berlin, 5th edition, 2017.

\bibitem{Dirac52}
G.~A. Dirac.
\newblock Some theorems on abstract graphs.
\newblock {\em Proc. London Math. Soc. (3)}, 2:69--81, 1952.

\bibitem{MR0337686}
C.~S. Edwards.
\newblock Some extremal properties of bipartite subgraphs.
\newblock {\em Canad. J. Math.}, 3:475--485, 1973.

\bibitem{ErdosG59}
P.~Erd{\H{o}}s and T.~Gallai.
\newblock On maximal paths and circuits of graphs.
\newblock {\em Acta Math. Acad. Sci. Hungar}, 10:337--356, 1959.

\bibitem{FederMotw10}
Tom\'{a}s Feder and Rajeev Motwani.
\newblock Finding large cycles in {H}amiltonian graphs.
\newblock {\em Discrete Appl. Math.}, 158(8):882--893, 2010.
\newblock \href {https://doi.org/10.1016/j.dam.2009.12.006}
  {\path{doi:10.1016/j.dam.2009.12.006}}.

\bibitem{FederMS02}
Tom\'{a}s Feder, Rajeev Motwani, and Carlos Subi.
\newblock Approximating the longest cycle problem in sparse graphs.
\newblock {\em SIAM J. Comput.}, 31(5):1596--1607, 2002.
\newblock \href {https://doi.org/10.1137/S0097539701395486}
  {\path{doi:10.1137/S0097539701395486}}.

\bibitem{fomin_et_al:LIPIcs:2019:11168}
Fedor~V. Fomin, Petr~A. Golovach, Daniel Lokshtanov, Fahad Panolan, Saket
  Saurabh, and Meirav Zehavi.
\newblock Going far from degeneracy.
\newblock {\em SIAM J. Discrete Math.}, 34(3):1587--1601, 2020.
\newblock \href {https://doi.org/10.1137/19M1290577}
  {\path{doi:10.1137/19M1290577}}.

\bibitem{FominGLPSZ21}
Fedor~V. Fomin, Petr~A. Golovach, Daniel Lokshtanov, Fahad Panolan, Saket
  Saurabh, and Meirav Zehavi.
\newblock Multiplicative parameterization above a guarantee.
\newblock {\em {ACM} Trans. Comput. Theory}, 13(3):18:1--18:16, 2021.
\newblock \href {https://doi.org/10.1145/3460956} {\path{doi:10.1145/3460956}}.

\bibitem{FominGSS22}
Fedor~V. Fomin, Petr~A. Golovach, Danil Sagunov, and Kirill Simonov.
\newblock Algorithmic extensions of {D}irac's theorem.
\newblock In {\em Proceedings of the 2022 Annual ACM-SIAM Symposium on Discrete
  Algorithms (SODA)}, pages 406--416, 2022.
\newblock URL: \url{https://epubs.siam.org/doi/abs/10.1137/1.9781611977073.20},
  \href
  {https://arxiv.org/abs/https://epubs.siam.org/doi/pdf/10.1137/1.9781611977073.20}
  {\path{arXiv:https://epubs.siam.org/doi/pdf/10.1137/1.9781611977073.20}},
  \href {https://doi.org/10.1137/1.9781611977073.20}
  {\path{doi:10.1137/1.9781611977073.20}}.

\bibitem{FominGSS22esa}
Fedor~V. Fomin, Petr~A. Golovach, Danil Sagunov, and Kirill Simonov.
\newblock Longest cycle above erd{\H{o}}s-gallai bound.
\newblock In {\em 30th Annual European Symposium on Algorithms, {ESA} 2022,
  September 5-9, 2022, Berlin/Potsdam, Germany}, volume 244 of {\em LIPIcs},
  pages 55:1--55:15. Schloss Dagstuhl - Leibniz-Zentrum f{\"{u}}r Informatik,
  2022.
\newblock \href {https://doi.org/10.4230/LIPIcs.ESA.2022.55}
  {\path{doi:10.4230/LIPIcs.ESA.2022.55}}.

\bibitem{FominK13}
Fedor~V. Fomin and Petteri Kaski.
\newblock Exact exponential algorithms.
\newblock {\em Commun. {ACM}}, 56(3):80--88, 2013.
\newblock URL: \url{http://doi.acm.org/10.1145/2428556.2428575}, \href
  {https://doi.org/10.1145/2428556.2428575}
  {\path{doi:10.1145/2428556.2428575}}.

\bibitem{FominLS14}
Fedor~V. Fomin, Daniel Lokshtanov, Fahad Panolan, and Saket Saurabh.
\newblock Efficient computation of representative families with applications in
  parameterized and exact algorithms.
\newblock {\em J. {ACM}}, 63(4):29:1--29:60, 2016.
\newblock URL: \url{http://doi.acm.org/10.1145/2886094}, \href
  {https://doi.org/10.1145/2886094} {\path{doi:10.1145/2886094}}.

\bibitem{FurerR94}
Martin F{\"{u}}rer and Balaji Raghavachari.
\newblock Approximating the minimum-degree steiner tree to within one of
  optimal.
\newblock {\em J. Algorithms}, 17(3):409--423, 1994.
\newblock \href {https://doi.org/10.1006/jagm.1994.1042}
  {\path{doi:10.1006/jagm.1994.1042}}.

\bibitem{Gabow07}
Harold~N. Gabow.
\newblock Finding paths and cycles of superpolylogarithmic length.
\newblock {\em SIAM J. Comput.}, 36(6):1648--1671, 2007.
\newblock \href {https://doi.org/10.1137/S0097539704445366}
  {\path{doi:10.1137/S0097539704445366}}.

\bibitem{GabowN08}
Harold~N. Gabow and Shuxin Nie.
\newblock Finding a long directed cycle.
\newblock {\em ACM Transactions on Algorithms}, 4(1), 2008.

\bibitem{GabowNisaac08}
Harold~N. Gabow and Shuxin Nie.
\newblock Finding long paths, cycles and circuits.
\newblock In {\em Proceedings of the 19th International Symposium on Algorithms
  and Computation (ISAAC)}, volume 5369 of {\em Lecture Notes in Comput. Sci.},
  pages 752--763. Springer, 2008.
\newblock \href {https://doi.org/10.1007/978-3-540-92182-0\_66}
  {\path{doi:10.1007/978-3-540-92182-0\_66}}.

\bibitem{GargP16}
Shivam Garg and Geevarghese Philip.
\newblock Raising the bar for vertex cover: Fixed-parameter tractability above
  a higher guarantee.
\newblock In {\em Proceedings of the Twenty-Seventh Annual {ACM-SIAM} Symposium
  on Discrete Algorithms (SODA)}, pages 1152--1166. {SIAM}, 2016.
\newblock \href {https://doi.org/10.1137/1.9781611974331.ch80}
  {\path{doi:10.1137/1.9781611974331.ch80}}.

\bibitem{DBLP:journals/mst/GutinKLM11}
Gregory Gutin, Eun~Jung Kim, Michael Lampis, and Valia Mitsou.
\newblock Vertex cover problem parameterized above and below tight bounds.
\newblock {\em Theory of Computing Systems}, 48(2):402--410, 2011.

\bibitem{GutinIMY12}
Gregory Gutin, Leo van Iersel, Matthias Mnich, and Anders Yeo.
\newblock Every ternary permutation constraint satisfaction problem
  parameterized above average has a kernel with a quadratic number of
  variables.
\newblock {\em J. Computer and System Sciences}, 78(1):151--163, 2012.

\bibitem{GutinMnich22}
Gregory~Z. Gutin and Matthias Mnich.
\newblock A survey on graph problems parameterized above and below guaranteed
  values.
\newblock {\em CoRR}, abs/2207.12278, 2022.
\newblock \href {https://arxiv.org/abs/2207.12278} {\path{arXiv:2207.12278}},
  \href {https://doi.org/10.48550/arXiv.2207.12278}
  {\path{doi:10.48550/arXiv.2207.12278}}.

\bibitem{GutinP16}
Gregory~Z. Gutin and Viresh Patel.
\newblock Parameterized traveling salesman problem: Beating the average.
\newblock {\em {SIAM} J. Discrete Math.}, 30(1):220--238, 2016.

\bibitem{GutinRSY07}
Gregory~Z. Gutin, Arash Rafiey, Stefan Szeider, and Anders Yeo.
\newblock The linear arrangement problem parameterized above guaranteed value.
\newblock {\em Theory Comput. Syst.}, 41(3):521--538, 2007.
\newblock \href {https://doi.org/10.1007/s00224-007-1330-6}
  {\path{doi:10.1007/s00224-007-1330-6}}.

\bibitem{ImpagliazzoPZ01}
Russell Impagliazzo, Ramamohan Paturi, and Francis Zane.
\newblock Which problems have strongly exponential complexity.
\newblock {\em J. Computer and System Sciences}, 63(4):512--530, 2001.

\bibitem{DBLP:conf/wg/Jansen0N19}
Bart M.~P. Jansen, L{\'{a}}szl{\'{o}} Kozma, and Jesper Nederlof.
\newblock Hamiltonicity below {D}irac's condition.
\newblock In {\em Proceedings of the 45th International Workshop on
  Graph-Theoretic Concepts in Computer Science (WG)}, volume 11789 of {\em
  Lecture Notes in Computer Science}, pages 27--39. Springer, 2019.

\bibitem{KargerMR97}
David~R. Karger, Rajeev Motwani, and G.~D.~S. Ramkumar.
\newblock On approximating the longest path in a graph.
\newblock {\em Algorithmica}, 18(1):82--98, 1997.
\newblock \href {https://doi.org/10.1007/BF02523689}
  {\path{doi:10.1007/BF02523689}}.

\bibitem{KleinbergT06}
Jon~M. Kleinberg and {\'{E}}va Tardos.
\newblock {\em Algorithm design}.
\newblock Addison-Wesley, 2006.

\bibitem{Koutis08}
Ioannis Koutis.
\newblock Faster algebraic algorithms for path and packing problems.
\newblock In {\em Proceedings of the 35th International Colloquium on Automata,
  Languages and Programming (ICALP)}, volume 5125 of {\em Lecture Notes in
  Comput. Sci.}, pages 575--586. Springer, 2008.

\bibitem{KoutisW16}
Ioannis Koutis and Ryan Williams.
\newblock Algebraic fingerprints for faster algorithms.
\newblock {\em Commun. {ACM}}, 59(1):98--105, 2016.
\newblock URL: \url{http://doi.acm.org/10.1145/2742544}, \href
  {https://doi.org/10.1145/2742544} {\path{doi:10.1145/2742544}}.

\bibitem{LokshtanovNRRS14}
Daniel Lokshtanov, N.~S. Narayanaswamy, Venkatesh Raman, M.~S. Ramanujan, and
  Saket Saurabh.
\newblock Faster parameterized algorithms using linear programming.
\newblock {\em {ACM} Trans. Algorithms}, 11(2):15:1--15:31, 2014.
\newblock \href {https://doi.org/10.1145/2566616} {\path{doi:10.1145/2566616}}.

\bibitem{MahajanRS09}
Meena Mahajan, Venkatesh Raman, and Somnath Sikdar.
\newblock Parameterizing above or below guaranteed values.
\newblock {\em J. Computer and System Sciences}, 75(2):137--153, 2009.

\bibitem{MishraRSSS11}
Sounaka Mishra, Venkatesh Raman, Saket Saurabh, Somnath Sikdar, and C.~R.
  Subramanian.
\newblock The complexity of {K}{\"{o}}nig subgraph problems and above-guarantee
  vertex cover.
\newblock {\em Algorithmica}, 61(4):857--881, 2011.
\newblock \href {https://doi.org/10.1007/s00453-010-9412-2}
  {\path{doi:10.1007/s00453-010-9412-2}}.

\bibitem{Monien85}
B.~Monien.
\newblock How to find long paths efficiently.
\newblock In {\em Analysis and design of algorithms for combinatorial problems
  ({U}dine, 1982)}, volume 109 of {\em North-Holland Math. Stud.}, pages
  239--254. North-Holland, Amsterdam, 1985.
\newblock URL: \url{http://dx.doi.org/10.1016/S0304-0208(08)73110-4}, \href
  {https://doi.org/10.1016/S0304-0208(08)73110-4}
  {\path{doi:10.1016/S0304-0208(08)73110-4}}.

\bibitem{RobertsonS95b}
Neil Robertson and Paul~D. Seymour.
\newblock Graph minors .xiii. the disjoint paths problem.
\newblock {\em J. Comb. Theory, Ser. {B}}, 63(1):65--110, 1995.
\newblock \href {https://doi.org/10.1006/jctb.1995.1006}
  {\path{doi:10.1006/jctb.1995.1006}}.

\bibitem{Vishwanathan04}
Sundar Vishwanathan.
\newblock An approximation algorithm for finding long paths in hamiltonian
  graphs.
\newblock {\em J. Algorithms}, 50(2):246--256, 2004.
\newblock \href {https://doi.org/10.1016/S0196-6774(03)00093-2}
  {\path{doi:10.1016/S0196-6774(03)00093-2}}.

\bibitem{Williams09}
Ryan Williams.
\newblock Finding paths of length $k$ in ${O}^*(2^k)$ time.
\newblock {\em Inf. Process. Lett.}, 109(6):315--318, 2009.

\end{thebibliography}
\end{document}